\documentclass[aps,prl,10pt,twocolumn,superscriptaddress,floatfix,nofootinbib,showpacs,longbibliography]{revtex4-2}
\usepackage{amsthm}
\usepackage{amsmath,bm}
\usepackage{amssymb}
\usepackage{amsfonts}
\usepackage{graphicx}
\usepackage{txfonts}
\usepackage{xcolor}
\usepackage{float}
\usepackage{braket}
\usepackage[colorlinks=true,linkcolor=blue,citecolor=blue,urlcolor=blue]{hyperref}
\usepackage[capitalise]{cleveref}
\usepackage{bbm}
\usepackage{changes}
\usepackage{ragged2e}

\newcommand{\tr}{\text{Tr}}

\newcommand{\ignore}[1]{} 

\newcounter{SaveEqnCntr}

\usepackage[justification=justified, format=plain]{subcaption}
\usepackage[justification=raggedright]{caption}

\newcommand{\be}{\begin{equation}}
\newcommand{\ee}{\end{equation}}
\newcommand{\ba}{\begin{eqnarray}}
\newcommand{\ea}{\end{eqnarray}}

\newtheorem{theorem}{Theorem}
\newtheorem{corollary}{Corollary}
\newtheorem{definition}{Definition}
\newtheorem{proposition}{Proposition}

\newtheorem{lemma}{Lemma}

\begin{document}
	
\title{Operational certification of nonclassicality in arbitrary quantum states from few copies}

\author{Sudip Chakrabarty}
\email{sudip27042000@gmail.com}
\affiliation{S. N. Bose National Centre for Basic Sciences, Block JD, Sector III, Salt Lake, Kolkata 700 106, India}

\author{Bivas Mallick}
\email{bivasqic@gmail.com}
\affiliation{S. N. Bose National Centre for Basic Sciences, Block JD, Sector III, Salt Lake, Kolkata 700 106, India}

\author{Ananda G. Maity}
\email{anandamaity289@gmail.com}
\affiliation{School of Physical Sciences, Indian Institute of Technology Goa, Ponda 403401, Goa, India}

\author{A. S. Majumdar}
\email{archan@bose.res.in}
\affiliation{S. N. Bose National Centre for Basic Sciences, Block JD, Sector III, Salt Lake, Kolkata 700 106, India}

\begin{abstract}
States with negative Wigner functions constitute a fundamental nonclassical resource underlying quantum advantage. However, their experimental certification typically relies on reconstructing the full phase-space distribution, resulting in a prohibitive measurement overhead for arbitrary quantum states. In this letter, we overcome this limitation by introducing Wigner moments, a family of global phase-space quantities that admit an exact multicopy realization as parity expectation values, and are therefore directly measurable from only a modest number of state copies. This operational correspondence enables systematic hierarchies of detection criteria together with experimentally accessible lower bounds on the logarithmic Wigner negativity, and constitutes a genuine measure of nonclassicality. Numerical benchmarking of our protocol  reveals a substantial reduction in copy budget relative to conventional Wigner tomography, establishing Wigner moments as an efficient framework for certifying continuous-variable quantum resources such as genuine multipartite entanglement.
\end{abstract}
\maketitle
{{\emph{Introduction---}}}Continuous-variable (CV) quantum systems, particularly optical platforms, constitute a central architecture for modern quantum technologies due to their experimental accessibility, high-bandwidth operation, and intrinsic scalability \cite{review_braunstein,agarwal2012quantum,gerry2023introductory}. In such systems, non-classicality is generally certified through properties of various phase-space quasiprobability distributions. Among them, the Wigner function plays a distinguished role, providing a complete phase-space representation of quantum states while retaining many features of a classical probability distribution \cite{wigner1932quantum,kenfack2004negativity,cormick2006classicality}. 

The significance of the Wigner function can also be highlighted by its connection to the Gottesman-Knill theorem, which establishes that quantum circuits restricted to Clifford operations admit efficient classical simulation \cite{gottesman1998heisenberg}. This efficient simulability, however, generally breaks down in the presence of Wigner negativity which is therefore widely regarded as a necessary resource for quantum computational advantage \cite{eisert2012positivewigner}. More broadly, states having negative Wigner-functions demonstrate quantum advantages in a variety of information processing tasks, including quantum computational speedup \cite{galvao2005discrete}, quantum error correction \cite{niset2009no}, quantum state distillation \cite{giedke2002characterization,fiuravsek2002gaussian,eisert2002distilling}, and many more \cite{bartlett2002efficient,eisert2012positivewigner}.

Driven by the fundamental role of Wigner negativity in quantum information processing, significant effort has been devoted to its detection and quantification \cite{chabaud2021witnessing,cimini2020neural,leibfried1996experimental,smithey1993measurement,banaszek1999direct,lutterbach1997method,bertet2002direct,winkelmann2022direct,lv2017reconstruction,lvovsky2001quantum,fluhmann2020direct,xiang2022distribution,mari2011directly}.
Existing methods predominantly rely on quantum-state tomography followed by reconstruction of the Wigner function \cite{lvovsky2009continuous,d2003quantum,rahimi2011quantum,d2001quantum}. However, quantum-state tomography requires a number of state copies that scales rapidly with the system size \cite{haah2017sample,odonnell2016efficient}, making it prohibitively expensive for large and arbitrary quantum states. There exist witness based approaches as well~\cite{chabaud2021witnessing,mari2011directly}. However these witnesses are tailored to specific classes of states, typically developed via semi-definite programs that are difficult to develop for arbitrary large systems. This raises a fundamental question: \emph{Can Wigner negativity in arbitrary states be certified from only a few number of state copies, without reconstructing the Wigner function?} 

In this Letter, we answer this question affirmatively by introducing an operational moment-based framework for certifying Wigner negativity of arbitrary quantum states using only a modest number of state copies. We first establish an exact multicopy representation of arbitrary Wigner moments as expectation values of parity observables, thereby transforming the certification of Wigner negativity from reconstructing the full Wigner function to estimating a small set of experimentally accessible moments. Exploiting universal constraints obeyed by positive Wigner functions, we then derive systematic hierarchies of moment-based witnesses that certify Wigner negativity in arbitrary CV quantum states. 

We further show that the same moments can be utilized to construct a bonafide measure of nonclassicality, providing experimentally certifiable lower bounds on the logarithmic Wigner negativity. We benchmark direct moment estimation using classical-shadow protocols against conventional Wigner tomography for representative mixed Wigner-negative states, showing that our
approach leads to a significant reduction in the copy budget. Moreover, our approach achieves reliable certification with the same physical copy budget while providing significantly higher estimation accuracy than tomographic reconstruction. Our formalism admits direct application in certifying
resources such as genuine multipartite entanglement.

{{\emph{Preliminaries---}}}We consider an $N$-mode CV quantum system described by bosonic annihilation operators $\vec{a}\coloneqq(a_1,\dots,a_N)$ satisfying the canonical commutation relations $[a_m,a_{m'}]=[a_m^\dagger,a_{m'}^\dagger]=0$ and $[a_m,a_{m'}^\dagger]=\delta_{m,m'}\mathbb{I}$. The associated Hilbert space is $\mathcal{H}=\bigotimes_{m=1}^N \mathcal{H}_m$, where each $\mathcal{H}_m$ is an infinite-dimensional space spanned by the Fock basis $\{\ket{n}\}_{n\in\mathbb{N}}$, and physical states are represented by density operators $\rho \in \mathcal{D}(\mathcal{H})$. Introducing the quadrature operators $\hat{x}_m = ({a_m + a_m^\dagger})/{\sqrt{2}}, ~\hat{p}_m = ({a_m - a_m^\dagger})/{\sqrt{2}i}$ which satisfy $[\hat{x}_m,\hat{p}_{m'}]=i\delta_{m,m'}\mathbb{I}$ (throughout the letter, we consider $\hbar=1$), we obtain a $2N$-dimensional phase-space description with canonical coordinates $(\vec{x},\vec{p})$. In this representation, the quantum state $\rho$ can be described by the Wigner function \cite{cahill1969density,wigner_function_for_pedestrians}
\begin{equation}
W_{\rho}(\vec{x},\vec{p}) = \frac{1}{(2\pi)^N} \int d^N y \;
\bra{\vec{x}+\frac{\vec{y}}{2}} \rho \ket{\vec{x}-\frac{\vec{y}}{2}} 
\, e^{-i\vec{p}\cdot\vec{y}},
\end{equation}
which is real and normalized, i.e., $\int d{x}^N\, d{p}^N \; W(\vec{x},\vec{p}) = 1$. Unlike classical probability distributions, the Wigner function can attain negative values, a hallmark of nonclassicality \cite{kenfack2004negativity, chakrabarty2026measurementinducedremoteactivationnonclassicality, mari2011directly}, that constitutes a key resource for CV quantum information processing. The relevant and essential features pertaining to
the Wigner function are summarized in \hyperref[SM_A]{Supplemental Material (SM)~A}~\cite{SM}.

{{\emph{Wigner moments as phase-space observables---}}}To certify Wigner negativity from a few copies of an arbitrary quantum state, without reconstructing its Wigner function, we introduce a family of global phase-space quantities that encode the signatures of Wigner negativity while admitting an exact operational realization. We refer to them as the Wigner moments ($w_m$), defined as follows (their basic properties are summarized in \hyperref[SM_B]{SM~B}~\cite{SM}).
\begin{definition} \label{moments}
Let $W$ be the Wigner function of an $N$-mode quantum state. The $m$-th order $W$-moments $(w_m)$ with $m \in \mathbb{N}$ are defined as:
\begin{equation} 
\begin{split}
 w_m:=\int_{\mathbb{R}^{2N}} W^{m} \hspace{0.2cm} d^N x \, d^N p. \label{eq:wigner_moments}
\end{split}
\end{equation}
\end{definition}
Having introduced the Wigner moments, a natural question is whether these nonlinear functionals admit any operational interpretation. We now answer this question affirmatively by showing that every Wigner moment can be realized exactly as the expectation value of a parity observable acting on a finite number of copies of the quantum state.

Let us first recall the definition of the Wigner function in the complex phase-space picture:
$W_\rho(\vec{\alpha}) = \mathrm{Tr}\!\left[\rho \, \Delta (\vec{\alpha})\right],$
where $\alpha = (x+ip)/\sqrt{2}$, such that $dx \,dp \equiv 2 d^2\alpha$, $\Delta(\vec{\alpha}) =  \left({1}/{\pi}\right) D(\vec{\alpha}) \Pi D^\dagger(\vec{\alpha})$ is the displaced parity operator,  $\Pi = e^{-i\pi a^\dagger a}$ is the parity operator and $D(\vec{\alpha}) = e^{\alpha a^\dagger - \alpha^* a}$ is the displacement operator. Consider the multi-copy operator
\[
A_n
=2\int d^2\alpha\,
\Delta(\vec{\alpha})^{\otimes n}.
\] 
Consequently, the $n$-th Wigner moment can be written as
\begin{equation}
w_n
=\mathrm{Tr}
\!\left[
\rho^{\otimes n}A_n
\right].
\label{eq:wn_multicopy}
\end{equation}
The operators $\{A_n\}$ are Hermitian by construction. The exact structure of \(A_n\) can be obtained by analyzing its position-space kernel and exploiting a decomposition into collective and relative degrees of freedom (see \hyperref[SM_C]{SM~C} \cite{SM}). The resulting operator takes a remarkably simple and universal form, which we state in the following theorem.

\begin{theorem}
\label{theorem:moment_operator}
Let \(w_n\) denote the \(n\)-th Wigner moment of a single-mode quantum state \(\rho\). Then
\begin{equation}
w_n
=\frac{1}{n\pi^{\,n-1}}
\mathrm{Tr}
\!\left[
\rho^{\otimes n}
F_n^\dagger
\Bigl(
I\otimes
\Pi^{\otimes(n-1)}
\Bigr)
F_n
\right],
\label{eq:theorem_moment_operator}
\end{equation}
where \(F_n\) is any orthogonal transformation whose first output mode coincides with the normalized collective coordinate
\[
Q_0
=\frac1{\sqrt n}
\sum_{j=1}^{n}x_j,
\]
and \(\Pi\) denotes the single-mode parity operator.
\end{theorem}

The proof can be found in \hyperref[SM_C]{SM~C}~\cite{SM}. Eq.~(\ref{eq:theorem_moment_operator}) reveals a simple geometric structure underlying all Wigner moments. In the collective-coordinate basis, one mode remains unchanged while all \((n-1)\) relative modes undergo a parity reflection. Consequently, the entire hierarchy of Wigner moments is encoded in the response of an \(n\)-copy state to parity operations acting only on the relative degrees of freedom. The complexity of the phase-space integral is thus replaced by a remarkably simple multi-copy observable. 

{{\emph{Experimental viability---}}}Beyond its conceptual significance, Theorem~\ref{theorem:moment_operator} immediately yields an operational measurement scheme. Instead of reconstructing the Wigner function, one prepares \(n\) identical copies of the state, applies the interferometric transformation \(F_n\), and measures the joint parity of the relative modes. Explicitly,
\begin{equation}
w_n
=\frac{1}{n\pi^{\,n-1}}
\left\langle
I\otimes
\Pi^{\otimes(n-1)}
\right\rangle_{F_n\rho^{\otimes n}F_n^\dagger}.
\label{eq:operational_formula}
\end{equation}
Thus every moment acquires a direct operational interpretation as a measurable multi-copy parity expectation value. This establishes the fundamental link between phase-space nonclassicality and experimentally accessible observables, providing the operational foundation for the certification criteria developed in this Letter. Further details of
an experimentally implementable protocol are discussed in \hyperref[SM_C]{SM~C}~\cite{SM}. 

The moments $\{w_m\}$ encode global properties of the phase-space distribution. Importantly, when $W$ is nonnegative, the sequence $\{w_m\}$ is constrained by a hierarchy of inequalities inherited from positive measures. Any violation of these constraints therefore certifies Wigner negativity. Below, we derive three complementary families of such constraints.

{{\emph{Hierarchy of detection criteria---}}}\emph{1. Detection hierarchy using $\mathcal{L}_p$ norms (LP).} 
Our first hierarchy is rooted in fundamental $\mathcal{L}_p$-norm inequalities obeyed by positive phase-space distributions. Expressed in terms of Wigner moments, these relations yield the following theorem.

\begin{theorem}\label{theorem: holder}
For any positive Wigner function $W$ and every integer $n\geq 2$, the following inequality holds:
\begin{equation}
LP(n) := (w_n)^{n-2} - (w_{n-1})^{n-1} \ge 0 .
\end{equation}
\end{theorem}
Any violation of the above inequality implies Wigner negativity. The proof follows from H\"older's inequality in $\mathcal{L}_p$ spaces, applied to the probability measure $d\mu = W\,dx^Ndp^N$ (valid when $W \ge 0$); see \hyperref[SM_D]{SM~D}~\cite{SM} for details.\qed

While Theorem \ref{theorem: holder} establishes a systematic set of constraints relating consecutive moments, each inequality involves only two moments. This motivates the search for stronger consistency conditions involving multiple moments simultaneously.

\emph{2. Detection hierarchy using log-convexity (LC) condition.}
Beyond the $\mathcal{L}_p$-norm hierarchy, positivity of the Wigner function also imposes log-convexity constraints on the moment sequence. This yields the following condition.
\begin{theorem} \label{theorem: cauchy}
Let $W$ be a positive Wigner function. Then the moments $\{w_n\}$ defined in Eq. \eqref{eq:wigner_moments} satisfy
\begin{equation}
LC(n) = w_{n-1} w_{n+1} -w_n^2 \ge 0.
\end{equation}
\end{theorem}
The proof of this Theorem can be found in SM D~\cite{SM}. All the hierarchies proposed so far do not necessarily provide strictly stronger criteria for detecting Wigner negativity, as they are based on a restricted subset of moments rather than the complete moment sequence and hence fail to fully capture the information contained in the complete moment sequence.

\emph{3. Moment-matrix approach via Hankel positivity.} A stronger approach is to consider all moments up to a given order simultaneously through Hankel matrices (moment matrices), which play a central role in classical moment problems (see \hyperref[SM_D]{SM~D}~\cite{SM} for a brief discussion). This leads to a hierarchy of moment-matrix based criteria that strictly strengthen the previously obtained moment inequalities.

\begin{theorem} \label{theorem:wigner_theorem}
For any positive Wigner function $W$ associated with a $N$-mode CV state $\rho$, the Hankel matrix $H_n(\boldsymbol{w})$ with elements
\begin{equation}
    [H_n(\boldsymbol{w})]_{ij} = w_{i+j+1},
    \qquad i,j \in \{0,1,\dots,n\}.
    \label{eq:wigner_hankel}
\end{equation}
is positive semidefinite for all $n \in \mathbb{N}$.
\end{theorem}

For the proof of Theorem~\ref{theorem:wigner_theorem}, we refer to \hyperref[SM_D]{SM~D}~\cite{SM}. Since these Hankel matrices are naturally nested,  each higher-order matrix incorporates all lower-order moments. This nested structure leads to a monotonic hierarchy of constraints, where higher-order conditions are strictly stronger than the lower-order ones. The following corollary formalizes this property.

\begin{corollary}[Monotonicity of the Hankel hierarchy]
\label{coro:monotonicity}
For the Hankel matrices $H_n(\boldsymbol{w})$ defined in Eq.~\eqref{eq:wigner_hankel},
\begin{align}
H_n(\boldsymbol{w}) \succeq 0 \;&\Rightarrow\; H_m(\boldsymbol{w}) \succeq 0,\quad \forall \; m <n, \\  
H_n(\boldsymbol{w}) \not\succeq 0 \;&\Rightarrow\; H_m(\boldsymbol{w}) \not\succeq 0,\quad \forall\, m\ge n.
\end{align}
\end{corollary}

Thus the Hankel criteria form a strictly increasing hierarchy of constraints, where higher-order matrices provide progressively stronger conditions for detecting Wigner negativity,
\begin{equation}
    H_1 \preceq H_2 \preceq H_3 \, \cdots.
\end{equation}

The lowest-order members of the three hierarchies coincide, yielding the criterion $\Gamma_{3} =w_2^2-w_3>0$ for the detection of Wigner negativity. This criterion is consistent with the previously established third-moment witness of Ref.~\cite{mallick_2025}.

To assess the performance of the proposed criteria, we analyze mixed Fock (Eq.~\ref{fock_mixed}) and noisy cat states (Eq.~\ref{cat_mixed}) in Table~\ref{tab:fock_mixture} and Fig.~\ref{fig:cat}, respectively. We compare the Wigner-negative regions detected by the LP, LC, and Hankel hierarchies for the above states. In both cases, the Hankel criteria exhibit a systematic improvement with increasing order and detect substantially larger portions of the nonclassical region, eventually recovering almost the full negativity range. A general comparison of the strength of different criteria and their detection performance is provided in the End Matter (see SM D~\cite{SM} for further discussion).

\begin{table}[t]
\centering
\begin{tabular}{c c c c c}
\hline\hline
$\quad n \quad $&
$\quad LP(2n+1)<0 \quad$ &
$\quad LC(2n)<0\quad$ &
$\quad H_n \not\succeq 0\quad $ \\
\hline
1 & $0<\lambda<0.309$ & $0<\lambda<0.309$ & $0 < \lambda < 0.309$ \\
2 & $0<\lambda<0.292$ & $0<\lambda<0.362$ & $0 < \lambda < 0.419$ \\
3 & $0<\lambda<0.294$ & $0<\lambda<0.364$ & $0 < \lambda < 0.469$ \\
4 & $0<\lambda<0.297$ & $0<\lambda<0.361$ & $0 < \lambda < 0.5$ \\
\hline\hline
\end{tabular}
\caption{ 
\justifying
\small
Detection ranges of $\lambda$ for Wigner negativity of the mixed Fock state: $\rho = \lambda |0\rangle\langle 0| + (1-\lambda) |1\rangle\langle 1|$, using different moment-based criteria. The parameter $n$ ensures that all criteria involve moments up to order $2n+1$. The Hankel-matrix conditions show a systematic and monotonic improvement with increasing $n$, eventually capturing the full nonclassical region, whereas the $\mathcal{L}_p$-norm (LP) and log-convexity (LC) based criteria lack such monotonicity, with higher-order conditions not necessarily yielding stronger detection.}
\label{tab:fock_mixture}
\end{table}

\begin{figure}
    \centering
        \includegraphics[width=0.95\linewidth]{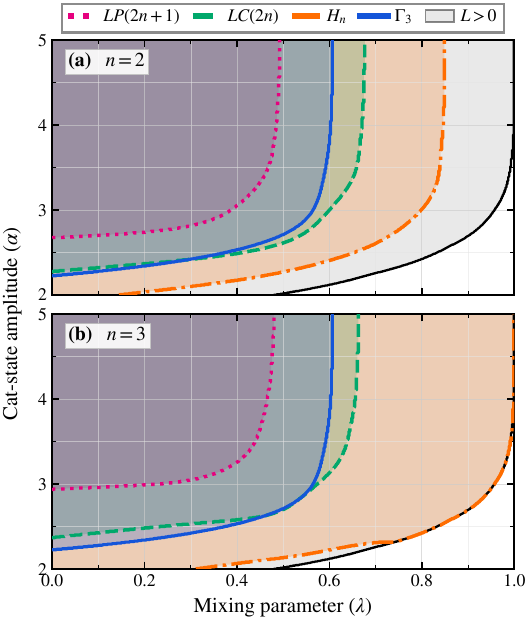}
    \caption{
\small
\justifying
Parameter regions exhibiting violations of different moment-based nonclassicality criteria for the mixed even cat state (see End matter for the details of this state). The coherent amplitude is taken in the range $\alpha\in[2,5]$, and $0\leq\lambda\leq1$ is the mixing parameter. 
(a) Comparison of \(LP(5)\), \(LC(4)\), and \(H_2\). (b) Comparison of \(LP(7)\), \(LC(6)\), and \(H_3\). The Hankel criteria detect the largest portion of the nonclassical region, followed by the log-convexity and \(\mathcal{L}_p\)-norm conditions. $L \equiv L(W)$ is the measure of nonclassicality.}
\label{fig:cat}
\end{figure}

{{\emph{Operational quantifier of Wigner negativity---}}}While the criteria developed above provide a practical and efficient way to identify Wigner negativity, they do not directly quantify its magnitude. Below we show that the same moments can also be used to yield certifiable lower bounds on the logarithmic Wigner negativity, a widely used measure of phase-space nonclassicality~\cite{PhysRevA.98.052350}. It is defined as:
\begin{equation}
L(W)
=\log \,\!\left(
\int |W(\vec{x},\vec{p})|
\, d^N x\, d^N p
\right),
\label{eq:logWignerNeg}
\end{equation}
which measures the total absolute volume contained in the Wigner distribution. However, evaluating Eq.~\eqref{eq:logWignerNeg} generally requires knowledge of the full Wigner function, making direct experimental estimation challenging. To connect logarithmic Wigner negativity with the moment framework developed above, we introduce the absolute-valued Wigner moments
\begin{equation}
\tilde{w}_r
=\int
|W(\vec{x},\vec{p})|^r
\, d^N x\, d^N p,
\qquad r\in\mathbb{N}.
\label{eq:absolute_moments}
\end{equation}
These quantities contain information about the distribution of negativity in phase space and satisfy $L(W)=\log(\tilde{w}_1)$. Now we define the family of moment-based quantities
\begin{equation}
L_r(W)
=\frac{1}{2-r}
\Big[
\log(\tilde{w}_r)
+
(1-r)\log(\tilde{w}_2)
\Big],
\quad r>2.
\label{eq:Lr}
\end{equation}
An important feature of this construction is that for every even integer \(r=2n\),
$\tilde{w}_{2n}=w_{2n}$.
Consequently, the quantities \(L_{2n}(W)\) can be evaluated directly from the experimentally accessible Wigner moments. Moreover, these quantities can be directly obtained from the same measurement data obtained while detecting Wigner negativity. The significance of \(L_r(W)\) is established by the following theorem.
\begin{theorem}[Moment-based lower bounds] \label{theorem:lower_bound_of_L}
For every quantum state and every \(r>2\),
\begin{equation}
 L(W) \ge L_r(W).
\end{equation}
\end{theorem}
The proof is provided in \hyperref[SM_E]{SM~E}~\cite{SM}. The theorem shows that low-order Wigner moments alone are sufficient to construct experimentally accessible lower bounds to  the logarithmic Wigner negativity. The quantities \(L_r(W)\) also retain a direct witnessing interpretation.

\begin{corollary}[Quantitative witness]\label{corollary:L_r_as_witness}
If $L_r(W)>0$, then the corresponding Wigner function must be negative in some region of phase space.
\end{corollary}

Thus every positive value of \(L_r(W)\) simultaneously certifies Wigner negativity and provides a nontrivial lower bound on the total negativity quantified by \(L(W)\). Finally, among all experimentally accessible members of this hierarchy, the lowest-order quantifier yields the strongest certifiable bound.

\begin{proposition}[Optimality of \(L_4\)]\label{proposition:optimality_of_L_4}
\begin{equation}
L_4(W)\ge L_r(W), \qquad \forall \; r>4.
\end{equation}
\end{proposition}
Therefore \(L_4(W)\) provides the tightest lower bound to logarithmic Wigner negativity obtainable from experimentally accessible even ordered quantifiers defined in Eq.~\eqref{eq:Lr}. 
The proofs of Corollary~\ref{corollary:L_r_as_witness} and Proposition~\ref{proposition:optimality_of_L_4} along with further properties of $L_r$, signifying it as a genuine measure of nonclassicality,  can be found in \hyperref[SM_E]{SM~E}~\cite{SM}.


 {{\emph{Advantage in nonclassicality certification using fewer copies---}}}The moment-based framework is designed to certify Wigner negativity without reconstructing the complete phase-space distribution. A natural operational question is whether these moments can indeed be estimated with substantially fewer state copies than conventional Wigner tomography. Since the moments \(w_n\) are expressed as multicopy nonlinear functions (see Eq.~\eqref{eq:wn_multicopy}) of the density operator, they can be estimated efficiently using classical-shadow techniques \cite{huang2020predicting}. To address the question on copy efficiency, we numerically benchmark the shadow estimator against direct Wigner tomography under the same total copy budget for two classes of single mode noisy states. In the tomographic approach, the available copies are distributed over a phase-space grid, displaced-parity measurements are used to reconstruct the sampled Wigner function, and then $\Gamma_3=w_2^2-w_3$ is evaluated only after reconstruction. By contrast, the shadow protocol estimates the same nonlinear functions directly from randomized measurements using cross-fitted U-statistics, completely bypassing the reconstruction of \(W(x,p)\).

We compare the copy efficiency with Wigner tomography in Fig.~\ref{fig:shadow_vs_tomo} (further numerical details are provided in the End Matter). The comparison reveals a clear copy advantage of moment based certification. With only \(10^4\) copies, the shadow estimator already reaches a mean additive error below \(14\%\) of the exact value of \(|\Gamma_3|\) for both representative states. The corresponding errors are approximately \(12.3\%\) and \(13.3\%\) for the mixed Fock and mixed odd cat states (see End matter for the details of these states), whereas Wigner tomography still exhibits errors of approximately \(339\%\) and \(5089\%\), respectively. Even after increasing the copy budget to \(5\times10^4\), tomography remains significantly less accurate, with errors of approximately \(62\%\) and \(390\%\), while the shadow estimator has already reached the \(5\%\) level in both examples. These results demonstrate that the observed copy advantage arises from directly estimating the experimentally relevant Wigner moments rather than reconstructing the complete phase-space distribution, thereby providing an operational justification for the few-copy paradigm advocated in this Letter.

Further, Fig.~\ref{fig:wigner_shadow} demonstrates that the shadow protocol efficiently recovers both the detection witness \(\Gamma_3\) and the quantitative bound \(L_4\) from only \(500\) classical shadows for Fock states $( n=1,2,\dots,10)$, highlighting the copy efficiency of the moment based approach.

\begin{figure}[t]
    \centering
    \includegraphics[width=0.95\linewidth]{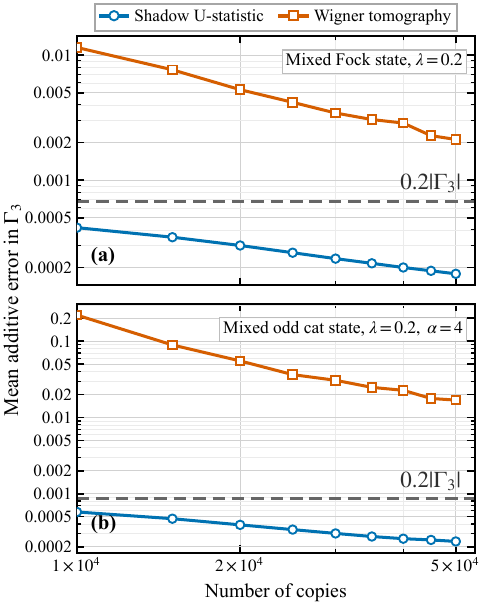}
    \caption{
\small
\justifying
The mean additive error in \(\Gamma_3\) obtained from the shadow estimator and from direct Wigner tomography, using the same total number of state copies. The dashed line marks \(20\%\) of the exact value of \(|\Gamma_3|\). The shadow estimator reaches this accuracy scale with far fewer copies $(\mathcal{N}<10^4 )$, whereas Wigner tomography does not reach it even with $5\times10^4$ copies.
}
\label{fig:shadow_vs_tomo}
\end{figure}

\begin{figure}
    \centering
    \includegraphics[width=0.9\linewidth]{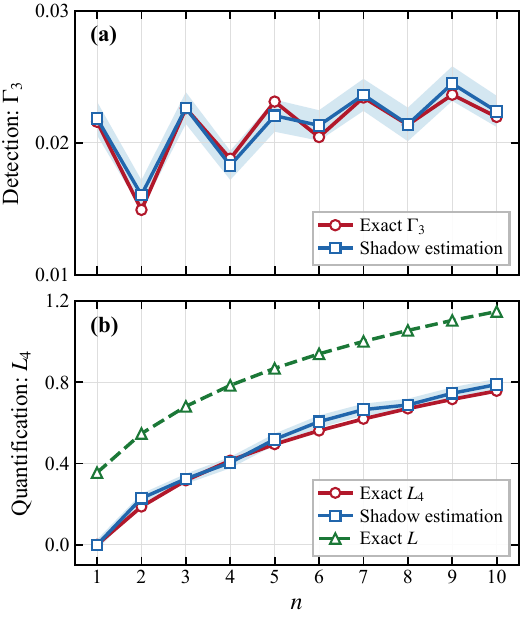}
   \caption{
\small
\justifying
Classical-shadow estimation for Fock states \(|n\rangle\). 
(a) \emph{Detection.} \(\Gamma_3\) obtained using $500$ classical shadows compared with the exact values. 
(b) \emph{Quantification.} A lower bound on the logarithmic Wigner negativity ($L$) is provided by $L_4$. Points show mean shadow estimates and shaded bands represent $95\%$ confidence intervals from \(100\) independent realizations.
}
\label{fig:wigner_shadow}
\end{figure}


{ \emph{ Conclusions---}}In this work, we have developed an operational framework for certifying 
Wigner negativity in arbitrary CV quantum states from a remarkably modest number of state copies, using global phase-space quantities called the Wigner moments. We have established that each Wigner moment admits an exact multicopy realization as the expectation value of a multi-copy parity observable, thereby transforming nonlinear functionals of the Wigner function into directly measurable observables. A bonafide and experimentally measurable quantifier of nonclassicality is constructed using the above
Wigner moments, providing a lower bound on the logarithmic Wigner negativity.

The formalism developed here shifts the problem of certification of Wigner negativity from reconstruction of the Wigner function in phase space to directly estimating a small set of phase-space observables. Benchmarking against conventional Wigner tomography demonstrates that direct moment estimation accomplishes the same certification task with substantially fewer state copies, thereby establishing a clear operational advantage over full state reconstruction.

Our results establish Wigner moments as a unifying operational framework for phase-space nonclassicality, bridging mathematical positivity constraints, quantitative resource measures, and experimentally accessible observables. From a broader perspective, since negativity of Wigner functions serves as a resource witness in diverse realms of quantum information~\cite{booth_contextuality,eisert2002distilling,eisert2012positivewigner}, the present approach can be naturally applied to provide a general route 
for experimentally certifying a wide range of resources directly through moments, including the certification of bipartite NPT entanglement~\cite{Zaw2024_bipartite_entanglement} and genuine multipartite entanglement~\cite{zaw2025_GME_1,zaw2025_GME_2} (see \hyperref[SM_F]{SM~F}~\cite{SM} for illustration). 

Our results open several promising directions for future research. On the theoretical side, an important challenge is to identify optimal and complete moment characterizations of Wigner-positive states and, more generally, to understand how the full hierarchy of Wigner moments captures the geometry of the Wigner-positive set. On the operational side, it will be important to establish the fundamental copy complexity of moment-based certification, including its scaling with the moment order and Hilbert-space truncation. 
Finally, developing statistically robust and loss-tolerant moment-based certification protocols for realistic experimental platforms represents an inevitable extension of our present analysis. 

{\emph{Acknowledgements.--}}
BM acknowledges DST INSPIRE fellowship program for financial support. AGM acknowledges the financial support provided by the IIT Goa Startup Grant (No. 2025/SG/AGM/058).


\newpage
\clearpage
\appendix
\begin{center}
    \textbf{\large{END MATTER}}
\end{center}
\vspace{0.2cm}
\begin{center}
    \textbf{Appendix I: Wigner function for Fock and Cat states}
\end{center}

For completeness, we collect here the explicit Wigner functions of the representative states used as examples throughout the main text. 

\textbf{1. Fock states.} The Fock states of a harmonic oscillator are defined as $|m\rangle \equiv \frac{(a^\dagger)^m}{\sqrt{m!}} |0\rangle$ with Wigner function, 
\begin{equation}
    W_{|m\rangle}(x,p)
    =
    \frac{(-1)^m}{\pi}
    e^{-(x^2+p^2)}
    \mathcal{L}_m\!\left(2(x^2+p^2)\right), \nonumber
\end{equation}
where $\mathcal{L}_m(z)$ denotes the $m$-th Laguerre polynomial. While the vacuum state $|0\rangle$ is Gaussian and therefore Wigner positive, all states with $m\ge1$ are Wigner negative. Consistent with this fact, we find that already the lowest-order Hankel condition detects nonclassicality of the Fock states, with $\det[H_1] < 0$ for all $m \ge 1$. Now we consider a mixed state,
\begin{equation}\label{fock_mixed}
    \rho = \lambda |0\rangle\langle 0| + (1-\lambda) |1\rangle\langle 1| ,
\end{equation}
with $0 \le \lambda \le 1$. The associated Wigner function is
\begin{equation}
    W(x,p)
    =
    \frac{1}{2\pi}
    \Big[(1-\lambda)(x^2+p^2) + 2\lambda - 1\Big]
    e^{-\frac{1}{2}(x^2+p^2)} . \nonumber
\end{equation}
The state is Wigner negative for $0<\lambda<0.5$. We compare the different moment-based criteria by determining the range of $\lambda$ over which negativity is detected; the results are summarized in Table~\ref{tab:fock_mixture}.

\textbf{2. Cat states.} 
Schr\"odinger cat states, defined as the superpositions of two coherent states, $\ket{\mathrm{cat}_{\pm} (\alpha)} = \frac{1}{\mathcal{N}_{\alpha}} \left( \ket{\alpha} \pm \ket{-\alpha} \right)$, with $\mathcal{N}_{\alpha} = \sqrt{2(1 \pm e^{-2\alpha^2})}$, possess the Wigner function~\cite{walschaers_prx_quantum}
\begin{align}
   W_{\text{cat}_{\pm}} (x,p) = \frac{1}{2 \pi {\mathcal{N}_\alpha^2}} &\Bigg( e^{ -\frac{1}{2} \left\{ (x-\alpha)^2 +p^2 \right\}} +e^{ -\frac{1}{2} \left\{ (x+\alpha)^2 +p^2 \right\}} \nonumber\\
   & \pm \cos (\sqrt{2} \alpha x ) e^{ -\frac{1}{2} \left( x^2 +p^2 \right)} \Bigg) \nonumber
\end{align}
whose interference term gives rise to Wigner negativity, a clear signature of nonclassicality. In contrast, the classical mixture of the same coherent states, $\rho_{\mathrm{cl}} = \frac{1}{2} \left( \ket{\alpha}\bra{\alpha} + \ket{-\alpha}\bra{-\alpha} \right)$, although non-Gaussian, possesses the positive Wigner function
\begin{align}
W_{\mathrm{cl}}(x,p) = \frac{1}{4 \pi} \Bigg( e^{ -\frac{1}{2} \left\{ (x-\alpha)^2 +p^2 \right\}} +e^{ -\frac{1}{2} \left\{ (x+\alpha)^2 +p^2 \right\}} \Bigg). \nonumber
\end{align}
We consider the noisy cat state,
\begin{equation}\label{cat_mixed}
\rho_{\mathrm{mix}} = (1-\lambda)\,\ket{\mathrm{cat}_{\pm}}\bra{\mathrm{cat}_{\pm}} + \lambda\, \rho_{\mathrm{cl}},
\end{equation}
with $0\leq\lambda\le1$ and corresponding Wigner function
\begin{equation}
W_{\mathrm{mix}}(x,p) = (1-\lambda)\, W_{\mathrm{cat}}(x,p) + \lambda\, W_{\mathrm{cl}}(x,p). \nonumber
\end{equation}
We compare the moment-based criteria developed here by determining the range of $\lambda$ for which Wigner negativity is detected for $2\leq \alpha\leq 5$ in Fig.~\ref{fig:cat}.

\begin{center}
    \textbf{Appendix II: Performance of different moment hierarchies}
\end{center}

\begin{figure}[t]
    \centering
    \includegraphics[width=0.49\textwidth]{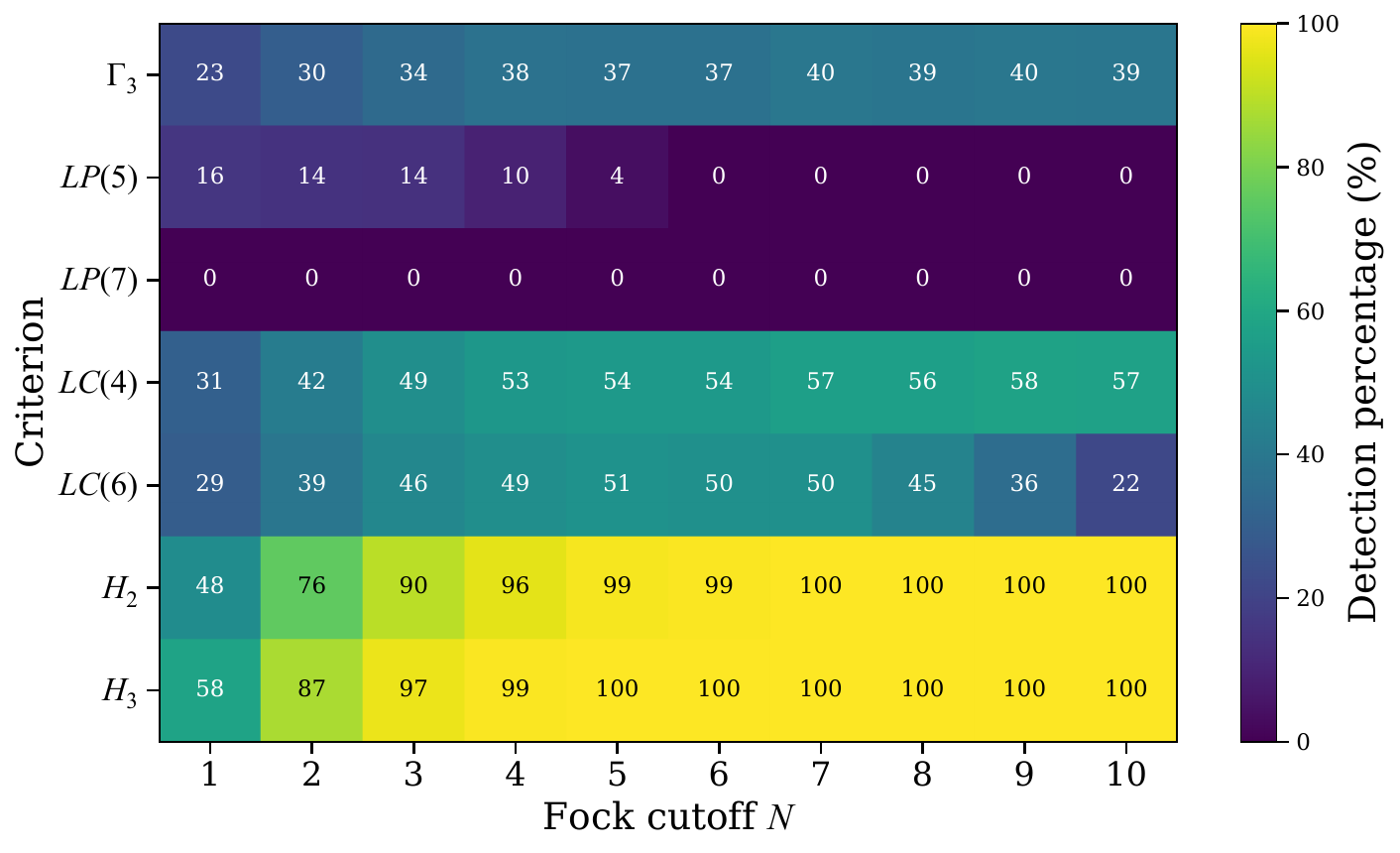}
\caption{
\small
\justifying
Detection probability of the moment-based hierarchies for random states drawn from the Hilbert-Schmidt ensemble in the truncated Fock space \(\mathrm{span}{|0\rangle,\ldots,|N\rangle}\), using \(10^4\) samples for each cutoff \(N\). Each cell shows the percentage of states detected by the corresponding criterion.}
\label{fig:performance}
\end{figure}

Here the main objective is to understand how efficiently the different criteria detect Wigner negativity when the maximum accessible moment order is fixed. To perform this analysis, we generate random quantum states in the truncated Fock space $\mathcal{H}_N = \mathrm{span}\left\{|0\rangle,|1\rangle,\dots,|N\rangle \right\}$, with cutoff values ranging from $N=1$ to $N=10$. The random states are sampled according to the Hilbert-Schmidt ensemble. Each density matrix is generated as
\begin{equation}
    \rho = \frac{GG^\dagger}{\text{Tr}(GG^\dagger)},
\end{equation}
where $G$ is a complex Ginibre matrix with independently distributed Gaussian entries. By construction, these states are positive semidefinite and properly normalized. For each random state, the Wigner function is computed numerically using the standard Fock-basis representation implemented in QuTiP. The Wigner moments $\{ w_n\}$ are then evaluated numerically through direct integration over phase space. A state is classified as Wigner negative whenever the minimum value of the Wigner function attains negative values on the numerical phase-space grid, below certain threshold $( -10^{-15})$. The fraction of randomly generated Wigner-negative states detected by each criterion is summarized in Fig.~\ref{fig:performance}. Several important features emerge from the numerical analysis. 

First, Wigner negativity rapidly becomes generic as the Fock cutoff increases. In particular, for $N\ge4$, essentially all randomly generated states exhibit Wigner negativity. This behavior reflects the fact that positivity of the Wigner function represents a constrained class within the full CV state space.

Second, the numerical results clearly demonstrate the hierarchy of detection strengths (see \hyperref[SM_D]{SM~D}~\cite{SM}),
\begin{equation}
\text{Hankel Positivity}
\Longrightarrow
\text{LC}
\Longrightarrow
\text{LP}.
\end{equation}
For every cutoff value, the Hankel hierarchy consistently detects the largest fraction of Wigner-negative states, followed by the LC hierarchy and finally the LP hierarchy.

Third, increasing the order of the scalar LP conditions does not improve the detection capability. This behavior indicates that isolated higher-order scalar inequalities capture only limited information about the underlying moment structure. A similar behavior is observed for the LC hierarchy. In contrast, the Hankel hierarchy becomes significantly stronger with increasing order. In particular, the higher-order Hankel conditions $H_2$ and $H_3$ detect almost all Wigner-negative states for sufficiently large cutoff values. 
This demonstrates that incorporating the global structure of the moment sequence through positivity of the full moment matrix is significantly more powerful than considering isolated scalar inequalities.

\begin{table}[t]
\begin{tabular}{ccccc}
\hline\hline
No. of copies
& \multicolumn{2}{c}{Mixed Fock state}
& \multicolumn{2}{c}{Mixed cat state} \\
\cline{2-3}\cline{4-5}
\(\mathcal{N}\)
& Shadow
& Wigner tomo.
& Shadow
& Wigner tomo. \\
\hline
$1 \times 10^4$ & 12.27\% & 338.88\% & 13.30\% & 5089.34\% \\
$2 \times 10^4$ &  8.83\% & 155.75\% &  9.01\% & 1275.82\% \\
$3 \times 10^4$ &  6.93\% & 101.34\% &  6.96\% &  710.30\% \\
$4 \times 10^4$ &  5.89\% &  84.19\% &  5.93\% &  524.66\% \\
$5 \times 10^4$ &  5.24\% &  62.25\% &  5.46\% &  390.19\% \\
\hline\hline
\end{tabular}
\caption{ 
\justifying
\small
Mean additive error in the witness \(\Gamma_3\), expressed as a percentage of the exact value \(|\Gamma_3|\). The shadow estimator directly estimates the moment witness using
U-statistics, whereas Wigner tomography first reconstructs the Wigner function on a phase-space grid and then evaluates \(\Gamma_3\) from the reconstructed grid.
\label{tab:copy_error_gamma3}
}
\end{table}

\begin{center}
    \textbf{Appendix III: Advantage of the Wigner moment formalism}
\end{center}

{{\textbf{1. Shadow estimation of Wigner moments.}}} The numerical implementation adapts the framework of Ref.~\cite{huang2020predicting} to phase-space observables. Since any single-mode CV state admits a Fock-basis representation
$\rho
=\sum_{m,n=0}^{\infty}
\rho_{mn}
|m\rangle\!\langle n|$,
and finite-energy states can be approximated efficiently 
by truncation to a sufficiently large Fock subspace, the simulations are performed in a truncated Fock basis containing the first \(d\) Fock states. For each copy of the quantum state, a unitary matrix \(U\in U(d)\) is sampled independently from the Haar measure. The use of Haar-random unitaries guarantees that every measurement basis is drawn uniformly from the unitary group, thereby implementing the standard global-measurement version of the classical-shadow protocol. After applying \(U\), a projective measurement is performed in the Fock basis, producing an outcome \(b\) with probability
\(
p(b|U)
=\langle b|U\rho U^\dagger|b\rangle.
\)
Each measurement outcome is converted into a classical shadow through the inverse measurement channel
\[
\hat{\rho}
=(d+1)\,
U^\dagger |b\rangle\!\langle b|U
-\mathbb{I},
\]
which satisfies
\(
\mathbb{E}[\hat{\rho}]
=\rho.
\)
Consequently, every shadow provides an unbiased estimator of the underlying quantum state. Repeating the procedure over \(N\) independent copies generates a collection of independent shadow snapshots $\{\hat{\rho}_1,\hat{\rho}_2,\ldots,\hat{\rho}_N\}$. To connect the shadow data with phase-space quantities, we precompute the Wigner kernel operators \(\Delta(\alpha)\) on a discretized phase-space grid. For each shadow and each phase-space point, we evaluate
\[
\hat{W_j}(\alpha)
=\mathrm{Tr}\!\left[\hat{\rho}_j\,\Delta(\alpha)\right].
\]
Because the shadow estimator is unbiased, one has
$\mathbb{E}\!\left[\hat W_j(\alpha)\right]
=
W(\alpha)$,
so that the collection of shadow-derived quasiprobabilities provides direct access to phase-space moments. Estimating higher-order moments requires products of several independent copies of the Wigner function. Rather than multiplying shadow estimates obtained from the same measurement record, which would introduce finite-sample bias, we employ unbiased U-statistic estimators. For a moment of order \(k\), the estimator is constructed from all distinct \(k\)-tuples of shadows. The extension to multimode systems is straight-forward. For an \(M\)-mode state, the phase-space kernel becomes
$\Delta(\boldsymbol{\alpha})
=\bigotimes_{j=1}^{M}
\Delta(\alpha_j)$,
and the corresponding moments are defined through integrals over the full \(2M\)-dimensional phase space. Since classical shadows naturally estimate operators acting on tensor-product Hilbert spaces, the same measurement and reconstruction protocol immediately yields unbiased estimators for multimode Wigner moments.

{{\textbf{2. Tomography versus shadow benchmark.}}} We now describe the numerical benchmark used to support the few-copy claim in the main text. For the direct Wigner-tomography baseline, we again use the discretized phase space grid large enough to contain the relevant support of the truncated state. The total copy budget is divided over the grid points. At each point, displaced-parity measurements provide samples of \(W(x,p)\). The moments \(w_2\) and \(w_3\) are then evaluated from the reconstructed grid. This gives tomography the same statistical care as the shadow estimator, while still forcing it to pay the cost of resolving the full phase-space distribution, so the accuracy of moment estimation is limited by the need to reconstruct much more information than is required for detecting negativity.

The benchmark was performed for two representative mixed Wigner-negative states: the mixed Fock state $\rho=0.2|0\rangle\langle0|+0.8|1\rangle\langle1|$,
and the noisy odd cat state
$\rho_{\rm mix}=0.8|\mathrm{cat}_{-}(\alpha)\rangle\langle\mathrm{cat}_{-}(\alpha)|
+0.2\,\rho_{\rm cl}$, with $\alpha=4$, embedded in a sufficiently large truncated Fock space $(d\leq 40)$. For each total copy number, we computed the mean additive error
$\mathbb{E}\left[|\widehat{\Gamma_3}-\Gamma_3|\right]$
over $500$ independent Monte-Carlo trials. The numerical results show that the shadow estimator reaches an error below \(0.14|\Gamma_3|\) already at \(\mathcal{N}=10^4\) copies for both states, where the tomographic estimation fails badly. 
The error percentages are summarised in Table~\ref{tab:copy_error_gamma3}. For the same physical copy budget, direct moment estimation gives a reliable negativity witness long before full Wigner tomography becomes considerable.

\clearpage
\onecolumngrid

\renewcommand{\theequation}{S\arabic{equation}}
\setcounter{equation}{0}
\noindent \begin{center}{\Large \bf Supplemental Material}\end{center}
~\vspace{-0.5cm}

\tableofcontents
\section{A. Quantum information in phase space: Wigner function}\label{SM_A}

Here, we briefly review the phase-space formulation of quantum mechanics based on Wigner function, which provides a powerful description of continuous-variable (CV) quantum systems with bosonic degrees of freedom, such as modes of the electromagnetic field. An $M$-mode system is described by annihilation operators $\vec{a} \coloneqq (a_1,\dots,a_M)$ satisfying the canonical commutation relations
\begin{equation}
    [a_m,a_{m'}]=[a_m^\dagger,a_{m'}^\dagger]=0, 
    \qquad 
    [a_m,a_{m'}^\dagger]=\delta_{m,m'}\mathbb{I}.
\end{equation}

The associated Hilbert space is $\mathcal{H}=\bigotimes_{m=1}^M \mathcal{H}_m$, where each $\mathcal{H}_m$ is an infinite-dimensional space spanned by the Fock basis $\{\ket{n}\}_{n\in\mathbb{N}}$, and physical states are represented by density operators $\rho \in \mathcal{D}(\mathcal{H})$. A convenient representation of CV systems is obtained by introducing quadrature operators
\begin{equation}
    \hat{x}_m = \frac{a_m + a_m^\dagger}{\sqrt{2}}, 
    \qquad 
    \hat{p}_m = \frac{a_m - a_m^\dagger}{\sqrt{2}i},
\end{equation}
which satisfy $[\hat{x}_m,\hat{p}_{m'}]=i\delta_{m,m'}\mathbb{I}$ (we set $\hbar=1$). These operators define a $2M$-dimensional phase space with canonical coordinates $(\vec{x},\vec{p})$. Equivalently, one may use the complex variables $\alpha_m = (x_m + i p_m)/\sqrt{2}$, which provide a compact description of phase space commonly used in quantum optics. The phase-space formulation of quantum mechanics represents quantum states through quasiprobability distributions. Among these, the Wigner function plays a central role. 
For a state $\rho$, it can be defined in the $(\vec{x},\vec{p})$ representation as
\begin{equation}
W_{\rho}(\vec{x},\vec{p}) = \frac{1}{(2\pi)^M} \int d^M y \;
\bra{\vec{x}+\frac{\vec{y}}{2}} \rho \ket{\vec{x}-\frac{\vec{y}}{2}} 
\, e^{-i\vec{p}\cdot\vec{y}},
\end{equation}
or, equivalently, in the complex phase-space picture via the displaced parity operator,
\begin{equation}
W_\rho(\vec{\alpha}) = \mathrm{Tr}\!\left[\rho \, \Delta (\vec{\alpha})\right],
\end{equation}
where $\Delta(\vec{\alpha}) =  \left(\frac{1}{\pi}\right)^MD(\vec{\alpha}) \Pi D^\dagger(\vec{\alpha})$, with $\Pi = \prod_m e^{-i\pi a_m^\dagger a_m}$ is the parity operator and $D(\vec{\alpha}) = \prod_m e^{\alpha_m a_m^\dagger - \alpha_m^* a_m}$ is the displacement operator. These two representations are completely equivalent and provide complementary perspectives: the $(x,p)$ formulation emphasizes the analogy with classical phase space, while the $\alpha$-representation highlights the role of displacement operations and is closely connected to experimental implementations. The Wigner function is real and normalized,
\begin{equation}
\int dx^M \, dp^M \; W(\vec{x},\vec{p}) = 1.
\end{equation}
Further, the Wigner function of a single mode system corresponding to a density matrix $\hat{\rho}$ can be expressed as
\begin{equation}
  W(x,p)=   \frac{1}{2 \pi} \int\limits_{ -\infty}^{\infty} \bra{x+\frac{y}{2}}\hat{\rho} \ket{x-\frac{y}{2}}e^{-ipy} \hspace{0.1cm} dy.
\end{equation}
For a pure state $\ket{\psi}$ with wave function $\psi(x)$, this Wigner function simplifies to 
\begin{equation}
W(x,p)=\frac{1}{2\pi} \int\limits_{-\infty}^{\infty}  \psi(x+\frac{y}{2})\psi^{*}(x-\frac{y}{2}) e^{-ipy} \hspace{0.1cm}dy,
\end{equation}
and its marginals reproduce measurable probability distributions,
\begin{align}
\int W(x,p)\,dp = |\psi(x)|^2, \quad
\int W(x,p)\,dx = |\phi(p)|^2,
\end{align}
where $\psi(x)$ and $\phi(p)$ denote the position and momentum wavefunctions, respectively. Despite the classical-like features, $W(\vec{x},\vec{p})$ can take negative values, reflecting the nonclassical nature of quantum states, such negativities cannot be reproduced by any classical joint probability distribution and are therefore regarded as a signature of nonclassicality.

Another key aspect of the phase-space approach is the correspondence between operators and phase-space functions via the Weyl transform. For an operator $\hat{A}$, one defines
\begin{equation}
\tilde{A}(x,p) = \int dy \; e^{-ipy} 
\bra{x+\frac{y}{2}} \hat{A} \ket{x-\frac{y}{2}},
\end{equation}
such that expectation values take the classical-like form
\begin{equation}
\langle \hat{A} \rangle = \int dx\,dp \; W(x,p)\,\tilde{A}(x,p).
\end{equation}
This mapping provides a direct bridge between operator algebra in Hilbert space and functions on phase space.

Within this framework, a fundamental distinction arises between Gaussian and non-Gaussian states. Gaussian states are characterized by Gaussian Wigner functions and therefore possess non-negative phase-space representations. Hudson's theorem states that a pure state has a non-negative Wigner function if and only if it is Gaussian \cite{hudson1974wigner,soto1983wigner}. This equivalence, however, does not extend to mixed states, for which non-Gaussian states may still exhibit positive Wigner functions \cite{mandilara2009extending,filip2011detecting}. Consequently, Wigner-negative states constitute a distinguished subclass of non-Gaussian states, and Wigner negativity is widely regarded as a key indicator of nonclassicality and a valuable resource for continuous-variable quantum information processing.

\section{B. The classical moment problem and its applications in quantum information theory}\label{SM_B}

The classical moment problem concerns the characterization of sequences that arise as moments of a positive measure. Given a sequence of real numbers $\{m_k\}_{k=0}^\infty$, the problem is to determine whether there exists a positive measure $\mu$ such that
\begin{equation}
    m_k = \int x^k \, d\mu(x), \quad \forall k \in \mathbb{N}.
\end{equation}
Sequences admitting such a representation are referred to as moment sequences. A central result in this context is that positivity of the underlying measure imposes strong consistency conditions on the sequence $\{m_k\}$. In particular, the associated Hankel matrices constructed from the moments must be positive semidefinite at all orders. More explicitly, defining the Hankel matrix $H_n(\mathbf{m})$ with elements
\begin{equation}
    [H_n(\mathbf{m})]_{ij} = m_{i+j}, \quad i,j \in \{0,1,\dots,n\},
\end{equation}
a necessary and sufficient condition for $\{m_k\}$ to correspond to a positive measure $\mu$ is
\begin{equation}
    H_n(\mathbf{m}) \geq 0 \quad \forall n.
\end{equation}
Equivalently, this condition can be understood as the requirement that all quadratic forms generated by polynomials are nonnegative,
\begin{equation}
    \int \left( \sum_{i=0}^k c_i x^i \right)^2 d\mu(x) \ge 0,
\end{equation}
for any real coefficients $\{c_i\}$. These constraints form a hierarchy of increasingly stringent conditions that any valid moment sequence must satisfy.

The framework of moment problems has found significant applications in quantum information theory, particularly in the characterization of quantum correlations and nonclassical features. One prominent example is the use of moments of the partially transposed density matrix, referred to as PT moments, for the detection of quantum entanglement \cite{elben2020mixed,yu2021optimal,Neven2021,miller2026detecting,mallick2025higher,mukherjee2025detecting, calabrese2012entanglement, tarabunga2025quantifyingmixedstateentanglementpartial}. 
Given a bipartite quantum state $\rho_{AB}$, the PT moments are defined as
\begin{equation}
    p_n = \text{Tr}\left[(\rho_{AB}^{T_B})^n\right], \label{PTmoments}
\end{equation}
for integers $n \geq 1$. These quantities provide indirect access to the spectrum of the partially transposed state $\rho^{T_B}_{AB}$, whose eigenvalues $\{\lambda_i\}$ determine entanglement properties. In particular, the characteristic polynomial
\begin{equation}
    \text{Det}(\rho^{T_B}_{AB} - \lambda I) = \sum_k a_k \lambda^k,
\end{equation}
has coefficients $a_k$ that are functions of the moments $\{p_n\}$. Since the transposition map is not physically implementable, these moments offer an experimentally accessible route to probe the spectrum.

A key advantage of moment-based approaches is their experimental accessibility. Techniques such as shadow tomography~\cite{aaronson2018shadow,huang2020predicting,cieslinski2024analysing} allow efficient estimation of moments without requiring full state reconstruction. 
These methods are particularly well-suited for higher-dimensional systems \cite{PhysRevA.109.012404}.
Moment-based methods have been further generalized to probe other forms of nonclassicality, including Margenau--Hill~\cite{chakrabarty2026quantifyingmargenauhillnonclassicality} and Kirkwood--Dirac nonpositivity~\cite{chakrabarty2026probing}, non-Markovianity \cite{mallick2024assessing} and quantum imaginarity \cite{chakrabarty2026detection}.

\section{C. Operator Representation and Accessibility of Wigner Moments}\label{SM_C}

In this section, we derive an exact operator representation of the Wigner moments. We show that the $n$-th Wigner moment can be expressed as the expectation value of a universal multi-copy observable acting on $n$ identical copies of the quantum state. This representation naturally reveals an experimentally feasible measurement scheme based on multi-mode interferometry and parity detection.

\subsection*{Multi-copy representation of Wigner moments}
The Wigner function of a single-mode quantum state $\rho$ can be expressed as
\begin{equation}
W(\alpha)
=\mathrm{Tr}\!\left[\rho\,\Delta(\alpha)\right],
\label{eq:wigner_operator_representation}
\end{equation}
where
\begin{equation}
\Delta(\alpha)
=\frac{1}{\pi}
D(\alpha)\Pi D^{\dagger}(\alpha).
\label{eq:phase_point_operator}
\end{equation}
denotes the displaced parity operator. Here,
\begin{equation}
D(\alpha)
=\exp\!\left(\alpha a^{\dagger}-\alpha^{*}a\right)
\end{equation}
is the displacement operator and
\begin{equation}
\Pi
=(-1)^{a^{\dagger}a}
\end{equation}
is the photon-number parity operator. Throughout this work, we adopt the phase-space convention
$\alpha=({q+ip})/{\sqrt2}$, such that $dq\,dp \equiv 2 d^2\alpha $. We freely switch between the notations
$W(\alpha)$ and $W(q,p)$ using this convention. In this convention, the vacuum state has the Wigner function
\begin{equation}
W_0(q,p)
=
\frac1\pi e^{-q^2-p^2},
\end{equation}
and the $n$-th Wigner moment is defined as
\begin{equation}
w_n
=2\int d^{2}\alpha\,W(\alpha)^n .
\label{eq:wn_definition}
\end{equation}
Substituting Eq.~(\ref{eq:wigner_operator_representation}) into Eq.~(\ref{eq:wn_definition}) yields
\begin{align}
w_n
&=
2\int d^{2}\alpha
\left[
\mathrm{Tr}\!\left(
\rho\Delta(\alpha)
\right)
\right]^n .
\label{eq:wn_step1}
\end{align}
Using the identity
\begin{equation}
\left[
\mathrm{Tr}(XY)
\right]^n
=\mathrm{Tr}
\!\left[
X^{\otimes n}
Y^{\otimes n}
\right],
\end{equation}
valid for arbitrary operators $X$ and $Y$, Eq.~(\ref{eq:wn_step1}) can be rewritten as
\begin{align}
w_n
&=
2\int d^{2}\alpha\,
\mathrm{Tr}
\!\left[
\rho^{\otimes n}
\Delta(\alpha)^{\otimes n}
\right]
\nonumber\\
&=
\mathrm{Tr}
\!\left[
\rho^{\otimes n}
\int 2d^{2}\alpha\,
\Delta(\alpha)^{\otimes n}
\right].
\end{align}
We therefore introduce the operator
\begin{equation}
A_n
=2\int d^{2}\alpha\,
\Delta(\alpha)^{\otimes n},
\label{eq:An_definition}
\end{equation}
which allows the $n$-th Wigner moment to be expressed as
\begin{equation}
w_n
=\mathrm{Tr}
\!\left[
\rho^{\otimes n}A_n
\right].
\label{eq:wn_operator_form}
\end{equation}

Equation~(\ref{eq:wn_operator_form}) provides a multi-copy representation of the Wigner moments. Rather than viewing $w_n$ as a nonlinear functional of the Wigner function, one may equivalently regard it as the expectation value of a state-independent observable $A_n$ acting on $n$ identical copies of the quantum state. Consequently, the problem of measuring Wigner moments reduces to determining the structure of the operator $A_n$.

\subsection*{Position-space kernel of the moment operator} 
To determine the operator $A_n$, we evaluate its kernel in the position representation. We begin with the position-space kernel of the displaced parity operator. In the convention adopted throughout this work, it is given by
\begin{equation}
\langle x|\Delta(q,p)|y\rangle
=
\frac{1}{\pi}
e^{ip(x-y)}
\delta(x+y-2q).
\label{eq:delta_kernel}
\end{equation}

Using Eq.~(\ref{eq:delta_kernel}), the kernel of $A_n$ can be written as
\begin{align}
\langle x_1,\ldots,x_n|
A_n
|y_1,\ldots,y_n\rangle
&=
\int dq\,dp
\prod_{j=1}^{n}
\langle x_j|
\Delta(q,p)
|y_j\rangle
\nonumber\\
&=
\frac1{\pi^n}
\int dq\,dp\,
e^{ip\sum_j(x_j-y_j)}
\prod_j\delta(x_j+y_j-2q).
\label{eq:kernel_step1}
\end{align}

The integration over the momentum variable can be performed immediately. Using
\begin{equation}
\int dp\,e^{ipS}
=
2\pi\,\delta(S),
\label{eq:p_integral}
\end{equation}
where
\begin{equation}
S
=
\sum_{j=1}^{n}(x_j-y_j),
\end{equation}
we obtain
\begin{align}
&
\langle x_1,\ldots,x_n|
A_n
|y_1,\ldots,y_n\rangle
=
\frac{2}{\pi^{n-1}}
\delta\!\left(
\sum_{j=1}^{n}(x_j-y_j)
\right)
\int dq
\prod_{j=1}^{n}
\delta(x_j+y_j-2q).
\label{eq:kernel_step2}
\end{align}

Using one of the delta functions to perform the $q$-integration yields
\begin{equation}
\int dq
\prod_{j=1}^{n}
\delta(x_j+y_j-2q)
=
\frac12
\prod_{j=2}^{n}
\delta\!\Big[
(x_j+y_j)-(x_1+y_1)
\Big].
\end{equation}

The remaining delta functions constrain all quantities $x_j+y_j$ to be equal. Explicitly,
\begin{equation}
x_1+y_1
= x_2+y_2=
\cdots
=x_n+y_n.
\label{eq:equal_sum_constraint}
\end{equation}
Let the common value of these expressions be denoted by $s$. Then
\begin{equation}
x_j=s-y_j,
\qquad
j=1,\ldots,n.
\label{eq:xj_relation1}
\end{equation}

Substituting Eq.~(\ref{eq:xj_relation1}) into the remaining constraint appearing in Eq.~(\ref{eq:kernel_step2}) gives
\begin{align}
0
&=
\sum_{j=1}^{n}(x_j-y_j)
\nonumber\\
&=
\sum_{j=1}^{n}(s-2y_j)
\nonumber\\
&=
ns-2\sum_{j=1}^{n}y_j.
\end{align}
Hence,
\begin{equation}
s
=
\frac{2}{n}
\sum_{j=1}^{n}y_j.
\label{eq:s_value}
\end{equation}

Substituting Eq.~(\ref{eq:s_value}) back into Eq.~(\ref{eq:xj_relation1}), we obtain
\begin{equation}
x_j
=
\frac{2}{n}
\sum_{k=1}^{n}y_k
-y_j,
\qquad
j=1,\ldots,n.
\label{eq:xj_final}
\end{equation}

It is convenient to collect these relations into a matrix equation. Defining
\begin{equation}
\mathbf{x}
=
\begin{pmatrix}
x_1\\
x_2\\
\vdots\\
x_n
\end{pmatrix},
\qquad
\mathbf{y}
=
\begin{pmatrix}
y_1\\
y_2\\
\vdots\\
y_n
\end{pmatrix},
\end{equation}
together with the matrix
\begin{equation}
M_n
=
\frac{2}{n}J-I,
\label{eq:Mn_definition}
\end{equation}
where $J$ denotes the $n\times n$ matrix whose entries are all equal to unity, Eq.~(\ref{eq:xj_final}) assumes the compact form
\begin{equation}
\mathbf{x}
=
M_n\mathbf{y}.
\label{eq:x_equals_My}
\end{equation}

Equation~(\ref{eq:x_equals_My}) reveals the geometric structure of the operator $A_n$. The kernel of $A_n$ is supported only on configurations related by the linear transformation $M_n$. In the following subsection, we show that this transformation is a reflection in the $(n-1)$-dimensional relative-coordinate subspace and derive the exact kernel of $A_n$, including the normalization factor arising from the associated Jacobian.

\subsection*{Reduction to a reflection operator} 
Equation~(\ref{eq:x_equals_My}) shows that the support of the kernel of $A_n$ is restricted to points satisfying the linear relation
\begin{equation}
\mathbf{x}
=M_n\mathbf{y},
\qquad
M_n
=
\frac{2}{n}J-I.
\end{equation}
To obtain the exact form of the kernel, it remains to determine the Jacobian associated with the transformation from the original set of delta-function constraints to the compact matrix constraint above.

Using the result of the previous subsection, the kernel of $A_n$ can be written as
\begin{equation}
\langle \mathbf{x}|A_n|\mathbf{y}\rangle
=
\frac1{\pi^{\,n-1}}
\delta\!\left(
\sum_j(x_j-y_j)
\right)
\int dq
\prod_j
\delta(x_j+y_j-2q).
\label{eq:kernel_before_jacobian}
\end{equation}
The constraints appearing in Eq.~(\ref{eq:kernel_before_jacobian}) can be expressed in terms of the functions
\begin{align}
g_1(\mathbf{x},\mathbf{y})
&=
\sum_{j=1}^{n}(x_j-y_j),
\label{eq:g1}
\\
g_i(\mathbf{x},\mathbf{y})
&=
x_i+y_i-x_1-y_1,
\qquad
i=2,\ldots,n.
\label{eq:gi}
\end{align}
The kernel therefore contains the product of delta functions
\begin{equation}
\delta(g_1)
\prod_{i=2}^{n}
\delta(g_i).
\label{eq:delta_constraints}
\end{equation}

As shown in the previous subsection, the simultaneous solution of these constraints is precisely
\begin{equation}
\mathbf{x}
=
M_n\mathbf{y}.
\label{eq:constraint_solution}
\end{equation}
The transformation from the variables \( (g_1,\ldots,g_n) \) to \( (x_1,\ldots,x_n) \) introduces a Jacobian factor. To evaluate it, we compute the matrix
\begin{equation}
J_g
=\left(
\frac{\partial g_i}{\partial x_j}
\right).
\end{equation}
Using Eqs.~(\ref{eq:g1}) and (\ref{eq:gi}), one finds
\begin{equation}
J_g
=\begin{pmatrix}
1 & 1 & 1 & \cdots & 1
\\
-1 & 1 & 0 & \cdots & 0
\\
-1 & 0 & 1 & \cdots & 0
\\
\vdots & \vdots & \vdots & \ddots & \vdots
\\
-1 & 0 & 0 & \cdots & 1
\end{pmatrix}.
\label{eq:jacobian_matrix}
\end{equation}

The determinant of this matrix can be evaluated conveniently using the Schur complement. Writing Eq.~(\ref{eq:jacobian_matrix}) in block form,
\begin{equation}
J_g
=\begin{pmatrix}
1 & \mathbf{1}^{T}
\\
-\mathbf{1} & I_{n-1}
\end{pmatrix},
\end{equation}
where $\mathbf{1}$ denotes the $(n-1)$-dimensional column vector with all entries equal to unity and $I_{n-1}$ is the $(n-1)\times(n-1)$ identity matrix, we obtain
\begin{align}
\det J_g
&=
\det(I_{n-1})
\,
\det\!\left(
1-\mathbf{1}^{T}I_{n-1}^{-1}(-\mathbf{1})
\right)
\nonumber\\
&=
1+\mathbf{1}^{T}\mathbf{1}
\nonumber\\
&=
1+(n-1)
\nonumber\\
&=
n.
\label{eq:jacobian_determinant}
\end{align}

The multidimensional delta-function identity therefore gives
\begin{equation}
\delta(g_1)
\prod_{i=2}^{n}
\delta(g_i)
=\frac{1}{|\det J_g|} \,
\delta(\mathbf{x}-M_n\mathbf{y}),
\end{equation}
which, together with Eq.~(\ref{eq:jacobian_determinant}), yields
\begin{equation}
\delta(g_1)
\prod_{i=2}^{n}
\delta(g_i)
=\frac{1}{n}
\delta(\mathbf{x}-M_n\mathbf{y}).
\label{eq:delta_reduction}
\end{equation}

Substituting Eq.~(\ref{eq:delta_reduction}) into Eq.~(\ref{eq:kernel_before_jacobian}), we finally obtain
\begin{equation}
\langle \mathbf{x}|A_n|\mathbf{y}\rangle
=
\frac1{n\pi^{\,n-1}}
\delta(\mathbf{x}-M_n\mathbf{y}).
\label{eq:kernel_final_result}
\end{equation}

Equation~(\ref{eq:kernel_final_result}) reveals that the action of $A_n$ is entirely determined by the matrix $M_n$. The operator maps a configuration $\mathbf{y}$ to the transformed configuration $M_n\mathbf{y}$ and is therefore generated by a linear orthogonal transformation in the $n$-dimensional configuration space. The structure of this transformation becomes particularly transparent after diagonalizing $M_n$. As we show in the next subsection, $M_n$ acts trivially on a single collective coordinate while reversing the sign of all relative coordinates. Consequently, $A_n$ can be interpreted as a reflection operator in the $(n-1)$-dimensional relative-coordinate subspace.

\subsection*{Exact operator form of $A_n$: Proof of Theorem~1} 
The structure of the kernel obtained in Eq.~(\ref{eq:kernel_final_result}) is entirely determined by the matrix
\begin{equation}
M_n
=\frac{2}{n}J-I.
\label{eq:Mn_recalled}
\end{equation}
To understand the action of the operator $A_n$, it is therefore necessary to analyze the spectral properties of $M_n$.

We begin by recalling that the matrix $J$ has rank one and can be written as
\begin{equation}
J
=
\mathbf{u}\mathbf{u}^{T},
\end{equation}
where
\begin{equation}
\mathbf{u}
=(1,1,\ldots,1)^{T}.
\end{equation}
Consequently,
\begin{equation}
J\mathbf{u}
=n\mathbf{u},
\end{equation}
while every vector orthogonal to $\mathbf{u}$ belongs to the kernel of $J$. Thus, the spectrum of $J$ consists of a single nonzero eigenvalue $n$ and $n-1$ vanishing eigenvalues,
\begin{equation}
\mathrm{spec}(J)
=\{n,0,\ldots,0\}.
\end{equation}

Since $M_n$ is related to $J$ through Eq.~(\ref{eq:Mn_recalled}), the corresponding eigenvalues are obtained immediately,
\begin{equation}
\mathrm{spec}(M_n)
=\left\{
\frac{2}{n}n-1,
\frac{2}{n}\times 0-1,
\ldots,
\frac{2}{n}\times 0-1
\right\},
\end{equation}
which simplifies to
\begin{equation}
\mathrm{spec}(M_n)
=\{1,-1,\ldots,-1\}.
\label{eq:spectrum_Mn}
\end{equation}

Equation~(\ref{eq:spectrum_Mn}) shows that $M_n$ possesses a unique eigendirection with eigenvalue $+1$, while every direction orthogonal to it acquires eigenvalue $-1$.

The distinguished eigenvector is
\begin{equation}
\mathbf{e}_0
=\frac{1}{\sqrt n}
(1,1,\ldots,1)^T,
\end{equation}
which naturally defines the collective coordinate
\begin{equation}
Q_0
=\mathbf{e}_0^{T}\mathbf{x}
=\frac{1}{\sqrt n}
\sum_{j=1}^{n}x_j.
\label{eq:collective_coordinate}
\end{equation}
The remaining $n-1$ coordinates are chosen as an arbitrary orthonormal basis spanning the subspace orthogonal to $\mathbf{e}_0$. Denoting these coordinates by
\begin{equation}
Q_1,Q_2,\ldots,Q_{n-1},
\end{equation}
the transformation from the original coordinates to the new coordinates may be written as
\begin{equation}
\mathbf{Q}
=F_n\mathbf{x},
\end{equation}
where $F_n$ is an orthogonal matrix whose first row is given by $\mathbf{e}_0^{T}$.

In this basis, the matrix $M_n$ becomes diagonal,
\begin{equation}
F_nM_nF_n^{T}
=\mathrm{diag}(1,-1,\ldots,-1).
\label{eq:Mn_diagonal}
\end{equation}

Equation~(\ref{eq:Mn_diagonal}) immediately reveals the geometric meaning of $M_n$. The collective coordinate remains unchanged,
\begin{equation}
Q_0
\longmapsto
Q_0,
\end{equation}
whereas every relative coordinate changes sign,
\begin{equation}
Q_j
\longmapsto
-Q_j,
\qquad
j=1,\ldots,n-1.
\end{equation}
Thus, $M_n$ acts as the identity on the one-dimensional collective subspace and as a reflection on the $(n-1)$-dimensional relative-coordinate subspace.

Using Eq.~(\ref{eq:Mn_diagonal}), the kernel in Eq.~(\ref{eq:kernel_final_result}) can be expressed in the transformed coordinates.
Since $F_n$ is orthogonal, the multidimensional Dirac delta distribution remains unchanged under the coordinate transformation. Consequently,
\begin{align}
&
\langle
Q_0,\ldots,Q_{n-1}
|
A_n
|
Q_0',\ldots,Q_{n-1}'
\rangle
\nonumber\\
&=
\frac{1}{n\pi^{n-1}}
\delta(Q_0-Q_0')
\prod_{j=1}^{n-1}
\delta(Q_j+Q_j').
\label{eq:kernel_Q}
\end{align}

The first factor in Eq.~(\ref{eq:kernel_Q}) is precisely the position-space kernel of the identity operator,
\begin{equation}
\langle Q|I|Q'\rangle
=\delta(Q-Q').
\end{equation}
The remaining factors can be identified with the position-space kernel of the parity operator,
\begin{equation}
\langle Q|\Pi|Q'\rangle
=\delta(Q+Q').
\label{eq:parity_kernel}
\end{equation}
Therefore,
\begin{equation}
\delta(Q_0-Q_0')
\prod_{j=1}^{n-1}
\delta(Q_j+Q_j')
=\left\langle
\mathbf{Q}
\left|
I\otimes\Pi^{\otimes(n-1)}
\right|
\mathbf{Q}'
\right\rangle .
\end{equation}

Substituting this identity into Eq.~(\ref{eq:kernel_Q}) yields the operator representation
\begin{equation}
A_n
=\frac{1}{n\pi^{\,n-1}}
F_n^{\dagger}
\left(
I\otimes
\Pi^{\otimes(n-1)}
\right)
F_n,
\label{eq:An_final}
\end{equation}
which proves Theorem~1.

Since $I\otimes\Pi^{\otimes(n-1)}$ acts identically on all relative modes, the resulting operator $A_n$ is independent of the specific orthonormal basis chosen for the relative-coordinate subspace. The theorem provides an exact operator representation of all Wigner moments. Remarkably, the structure of $A_n$ is universal: irrespective of the value of $n$, the operator consists of a parity measurement acting on the relative-coordinate subspace, while the collective mode contributes only an identity operator.

\subsection*{Experimental feasibility} 
The operator representation derived in Theorem~1 provides a direct experimental protocol for measuring arbitrary Wigner moments. In contrast to conventional approaches based on quantum state tomography, the present method allows the moments to be estimated directly from expectation values of a multi-copy observable. Consider the $n$-th Wigner moment,
\begin{equation}
w_n
=\frac{1}{n\pi^{\,n-1}}
\mathrm{Tr}
\!\left[
\rho^{\otimes n}
F_n^{\dagger}
\left(
I\otimes
\Pi^{\otimes(n-1)}
\right)
F_n
\right] = \frac{1}{n\pi^{\,n-1}}
\left\langle
I\otimes
\Pi^{\otimes(n-1)}
\right\rangle_{F_n\rho^{\otimes n}F_n^{\dagger}}.
\label{eq:wn_measurement}
\end{equation}
Equation~(\ref{eq:wn_measurement}) immediately suggests a three-step measurement protocol. 
\begin{enumerate}
    \item The first step consists of preparing $n$ identical copies of the quantum state. The observable associated with the $n$-th Wigner moment acts on the tensor-product Hilbert space
$\mathcal{H}^{\otimes n}$, and therefore requires simultaneous access to $n$ copies of the state $\rho$.

\item The second step is the implementation of the orthogonal interferometric transformation $F_n$. As shown in the previous subsection, the role of $F_n$ is to separate the collective degree of freedom from the relative degrees of freedom. In the transformed basis, the first mode corresponds to the collective coordinate
$Q_0
=\frac{1}{\sqrt n}
\sum_{j=1}^{n}x_j$, 
while the remaining modes span the relative-coordinate subspace. Physically, the transformation $F_n$ can be realized using a passive Gaussian interferometric network implementing the corresponding orthogonal transformation of the quadratures. Such networks may be decomposed into beam splitters and phase shifters according to standard linear-optical constructions.

\item The final step is the measurement of parity on the relative modes. Since the operator acting in the transformed basis is
$I\otimes
\Pi^{\otimes(n-1)}$,
the collective mode contributes only an identity operator and therefore does not affect the measurement outcome. All relevant information is contained in the relative modes. The measurement protocol therefore requires estimation of the expectation value of the joint parity observable
\begin{equation}
\Pi^{\otimes(n-1)}
=\Pi_1\Pi_2\cdots\Pi_{n-1},
\end{equation}
acting on the relative-coordinate subspace. Combining these three steps yields an estimator for the Wigner moment,
\begin{equation*}
w_n
=\frac{1}{n\pi^{\,n-1}}
\left\langle
I\otimes
\Pi^{\otimes(n-1)}
\right\rangle_{F_n\rho^{\otimes n}F_n^{\dagger}}.
\end{equation*}
\end{enumerate}

Consequently, the determination of the Wigner moments requires neither reconstruction of the Wigner function nor full quantum state tomography. Instead, each moment is obtained directly from a single expectation value measured after an interferometric processing of $n$ identical copies of the state.
The physical interpretation of this result is particularly transparent. The transformation $F_n$ isolates a single collective mode that remains invariant under the reflection generated by the operator $A_n$. All nontrivial information about the Wigner moment is encoded in the relative-coordinate subspace, where the action of $A_n$ reduces to a simultaneous parity measurement. The $n$-th Wigner moment therefore quantifies the response of the $n$-copy state to a reflection of all relative degrees of freedom while leaving the collective degree of freedom unchanged.

Theorem~1 thus establishes that the complete hierarchy of Wigner moments can, in principle, be accessed through multicopy interferometry followed by parity measurements on the relative modes. Since the protocol directly estimates the desired moment without reconstructing the underlying phase-space distribution, it provides an experimentally motivated alternative to full state tomography for characterizing phase-space properties of continuous-variable quantum states.

The experimental protocol described above requires the implementation of an orthogonal transformation $F_n$ that separates the collective degree of freedom from the relative degrees of freedom of the \(n\)-copy state. In general, many such transformations are possible. The only requirement is that the first output mode coincide with the normalized collective coordinate $Q_0$.
while the remaining output modes form an orthonormal basis of the (n-1)-dimensional relative-coordinate subspace. For low-order moments, particularly simple choices of $F_n$ may be constructed analytically. These examples illustrate how the collective and relative coordinates arise from the original position coordinates of the n-copies.

For $n=2$, we choose
\begin{equation}
F_2
\equiv\frac{1}{\sqrt2}
\begin{pmatrix}
1 & 1\\
1 & -1
\end{pmatrix}.
\end{equation}
The transformed coordinates are
\begin{equation}
Q_0
=\frac{x_1+x_2}{\sqrt2},
\qquad
Q_1
=\frac{x_1-x_2}{\sqrt2}.
\end{equation}
The coordinate $Q_0$ describes the collective motion of the two copies, while $Q_1$ measures their relative displacement. This transformation is precisely the one implemented by a balanced $50:50$ beam splitter.

For $n=3$, a convenient choice is
\begin{equation}
F_3
\equiv
\begin{pmatrix}
\dfrac1{\sqrt3} &
\dfrac1{\sqrt3} &
\dfrac1{\sqrt3}
\\[10pt]
\dfrac1{\sqrt2} &
-\dfrac1{\sqrt2} &
0
\\[10pt]
\dfrac1{\sqrt6} &
\dfrac1{\sqrt6} &
-\sqrt{\dfrac23}
\end{pmatrix}
\end{equation}
The corresponding coordinates are
\begin{align}
Q_0
=
\frac{x_1+x_2+x_3}{\sqrt3},
\quad
Q_1
=
\frac{x_1-x_2}{\sqrt2},
\quad
Q_2
=
\frac{x_1+x_2-2x_3}{\sqrt6}.
\end{align}

For $n=4$, we may choose
\begin{equation}
F_4
\equiv\
\begin{pmatrix}
\dfrac12 &
\dfrac12 &
\dfrac12 &
\dfrac12 \\[4mm]

\dfrac1{\sqrt2} &
-\dfrac1{\sqrt2} &
0 &
0 \\[4mm]

\dfrac12 &
\dfrac12 &
-\dfrac12 &
-\dfrac12 \\[4mm]

0 &
0 &
\dfrac1{\sqrt2} &
-\dfrac1{\sqrt2}
\end{pmatrix}
\end{equation}
The transformed coordinates become
\begin{align}
Q_0
=
\frac{x_1+x_2+x_3+x_4}{2},
\quad
Q_1
=
\frac{x_1-x_2}{\sqrt2},
\quad
Q_2
=
\frac{x_1+x_2-x_3-x_4}{2},
\quad
Q_3
=
\frac{x_3-x_4}{\sqrt2}.
\end{align}
In each case, the rows of $F_n$ are orthonormal and satisfy
\begin{equation}
F_nF_n^{T}
=I.
\end{equation}
Furthermore, the transformations diagonalize the reflection matrix ($M_n$),
\begin{equation}
F_nM_nF_n^{T}
=\mathrm{diag}(1,-1,\ldots,-1).
\end{equation}
Consequently, the first output mode is left invariant by the action of ($M_n$), whereas every relative mode acquires a sign inversion. This is precisely the structure responsible for the appearance of the operator
\begin{equation}
I\otimes\Pi^{\otimes(n-1)}
\end{equation}
in Theorem~1. The transformations $F_2$, $F_3$, and $F_4$ provide explicit examples
of collective-coordinate transformations entering the
moment-measurement protocol.

\section{D. Detection of Wigner negativity}\label{SM_D}
Here we give the detailed proofs of the three complementary hierarchy of detection criteria discussed in the main text.
\subsection{Proof of Theorem 2}
Here we aim to prove that for any positive normalized Wigner function $W$ and every integer $n\geq 2$, $(w_n)^{n-2} - (w_{n-1})^{n-1} \ge 0 $. Since $W \ge 0$ and $\int W \, dx^N dp^N= 1$, we define a probability measure
\begin{equation}
d\mu = W \, dx^N dp^N, \;\; \text{such that}\;\; \int d\mu =1.
\end{equation}
With respect to $\mu$, the moments can be written as:
\begin{equation}
w_n = \int W^{n-1} \, d\mu, \quad
w_{n-1} = \int W^{n-2} \, d\mu.
\end{equation}
We now consider $\mathcal{L}_p$ norms with respect to this probability measure, defined as 
\begin{equation}
\|f\|_p = \left( \int |f|^p \, d\mu \right)^{1/p}.
\end{equation}

\textit{Hölder's inequality in $\mathcal{L}_p$ space:}
Let $1 \leq p < \infty$ and $1 < q < \infty$ be such that $\frac{1}{p} + \frac{1}{q} = 1$. For functions $f \in \mathcal{L}_p$ and $g \in \mathcal{L}_q$, their product $fg$ belongs to $\mathcal{L}_1$. Furthermore, Hölder's inequality states that
\begin{equation}
     \|fg\|_{1} \leq \|f\|_{p} \, \|g\|_{q}. \label{eq: holder}
\end{equation}
As a special case, when $p = q = 2$, this inequality reduces to the well-known Cauchy-Schwarz inequality. Here we choose,
\begin{equation}
f = W^{n-2}, \quad g = 1,
\end{equation}
and exponents
\begin{equation}
p = \frac{n-1}{n-2}, \quad q = n-1.
\end{equation}
Then, using Eq.~\eqref{eq: holder} we get,
\begin{equation}
\|W^{n-2}\|_1 \le \|W^{n-2}\|_p \|1\|_q.
\end{equation}
We now evaluate each term;
\begin{align}
\|W^{n-2}\|_1 &= \int |W^{n-2}| \, d\mu \nonumber \\
& = \int W^{n-2} \, d\mu \qquad \text{(since $W\geq0$)} \nonumber\\
&= w_{n-1},\\
\|W^{n-2}\|_p
&= \left( \int |W^{n-2}|^p \, d\mu \right)^{1/p} \nonumber  \\
&= \left( \int W^{(n-2) \left( \frac{n-1}{n-2} \right)} \, d\mu \right)^{1/p} \nonumber  \\
&=\left( \int W^{n-1} \, d\mu \right)^{1/p} \nonumber \\
&=w_n^{1/p}.
\end{align}
Finally,
\begin{equation}
\|1\|_q = \left( \int  d\mu \right)^{1/q} = 1,
\end{equation}
since $\mu$ is a probability measure. Therefore,
\begin{equation}
w_{n-1} \le w_n^{1/p}.
\end{equation}
Substituting $p = \frac{n-1}{n-2}$, we obtain
\begin{equation}
w_{n-1} \le w_n^{\frac{n-2}{n-1}}.
\end{equation}
Raising both sides to the power $(n-1)$ gives
\begin{align}
&(w_{n-1})^{n-1} \le (w_n)^{n-2} \nonumber\\
\Rightarrow \;\;&(w_n)^{n-2} - (w_{n-1})^{n-1} \geq 0,
\end{align}
which completes the proof.

\subsection{Proof of Theorem 3}

With the measure $d\mu = W\, dx^N dp^N$, which is positive and normalized, we can write the moments as
\[
w_n = \int W^{n-1} d\mu.
\]
Applying the Cauchy--Schwarz inequality to the functions
\begin{align}
    f = W^{\frac{n-2}{2}}, \qquad g = W^{\frac{n}{2}},
\end{align}
we obtain
\begin{align*}
\int W^{n-1} d\mu
& \le
\left( \int W^{n-2} d\mu \right)^{\frac{1}{2}}
\left( \int W^n d\mu \right)^{\frac{1}{2}}.
\end{align*}
Taking square on both sides we get the set of conditions,
\begin{align}
&w_n^2 \le w_{n-1} w_{n+1} \; \\
\Rightarrow \quad&\; w_{n-1} w_{n+1} -w_n^2 \ge 0.
\end{align}
This completes the proof.

\subsection{Proof of Theorem 4}
We now aim to prove that for any positive normalized Wigner function $W$, the Hankel matrix $H_n(\boldsymbol{w})$ with elements, $[H_n(\boldsymbol{w})]_{ij} = w_{i+j+1}$, $i,j \in \{0,1,\dots,n\}$ is positive semidefinite for all $n \in \mathbb{N}$. Let $\boldsymbol{y} = (y_0, y_1, \dots, y_n) \in \mathbb{R}^{n+1}$ be an arbitrary
real vector. Consider the quadratic form associated with the Hankel matrix
$H_n(\boldsymbol{w})$:
\begin{equation}
\begin{split}
    \boldsymbol{y}^T H_n(\boldsymbol{w}) \boldsymbol{y}
    &= \sum_{i=0}^m \sum_{j=0}^m y_i y_j \, w_{i+j+1}.\\
    &= \sum_{i=0}^n \sum_{j=0}^n y_i y_j
    \int_{\mathbb{R}^{2N}} W^{\,i+j+1} \, d^N x \, d^N p .
\end{split}
\end{equation}

Since $W \ge 0$ and the summations are finite, Fubini's theorem allows us to interchange the order of summation and integration, yielding
\begin{equation}
\begin{split}
    \boldsymbol{y}^T H_n(\boldsymbol{w}) \boldsymbol{y}
    &= \int_{\mathbb{R}^{2N}}
    \left(
        \sum_{i=0}^n \sum_{j=0}^n
        y_i y_j W^{\,i+j+1}
    \right)
    d^N x\, d^N p .
\end{split}
\end{equation}

We now rewrite the integrand by factoring out a single power of $W$:
\begin{equation}
\begin{split}
    \sum_{i,j=0}^n y_i y_j W^{i+j+1}
    &= W
    \sum_{i,j=0}^N
    y_i y_j W^i W^j .
\end{split}
\end{equation}
The double sum appearing above can be recognized as a perfect square,
\begin{equation}
\begin{split}
    \sum_{i,j=0}^n
    y_i y_j W^i W^j
    =
    \left(
        \sum_{i=0}^n y_i W^i
    \right)^2 .
\end{split}
\end{equation}
Therefore, the quadratic form can be written as,
\begin{equation}
\begin{split}
    \boldsymbol{y}^T H_n(\boldsymbol{w}) \boldsymbol{y}
    &=
    \int_{\mathbb{R}^{2N}}
    W
    \left(
        \sum_{i=0}^n y_i W^i
    \right)^2
    d^N x\, d^N p .
\end{split}
\end{equation}

By assumption, $W \ge 0$ in all phase points, and the squared term is
manifestly nonnegative. Hence, the integrand is nonnegative pointwise over
phase space, implying
\begin{equation}
    \boldsymbol{y}^T H_n(\boldsymbol{w}) \boldsymbol{y} \ge 0
    \qquad \forall \ \boldsymbol{y} \in \mathbb{R}^{n+1}.
\end{equation}
This establishes that the Hankel matrix $H_n(\boldsymbol{w})$ is positive
semidefinite, completing the proof.

\subsection{Proof of Corollary 1}
We first prove the forward implication. For any $m < n$, the matrix $H_m(\boldsymbol{w})$ is a principal submatrix of $H_n(\boldsymbol{w})$, obtained by restricting to indices $i,j \in \{0,1,\dots,m\}$. Since principal submatrices of positive semidefinite matrices are themselves positive semidefinite, $H_n(\boldsymbol{w}) \succeq 0 \;\Rightarrow\; H_m(\boldsymbol{w}) \succeq 0 \quad \forall m<n$.

For the converse, if $H_n(\boldsymbol{w})\not\succeq0$, then there exists a vector $\boldsymbol{y}\in\mathbb{R}^{n+1}$ such that $\boldsymbol{y}^TH_n(\boldsymbol{w})\boldsymbol{y}<0$. Now considering any $m > n$, define a vector $\tilde{\boldsymbol{y}} \in \mathbb{R}^{m+1}$ by padding zeros, i.e., $\tilde{\boldsymbol{y}} := (y_0, y_1, \dots, y_n, 0, \dots, 0)$. Then, $\tilde{\boldsymbol{y}}^T H_m(\boldsymbol{w}) \tilde{\boldsymbol{y}} = \boldsymbol{y}^T H_n(\boldsymbol{w}) \boldsymbol{y} < 0$, since the additional rows and columns do not contribute. Hence, $H_m(\boldsymbol{w}) \not\succeq 0$ for all $m \geq n$.

\subsection{Comparison between the LP, LC and Hankel hierarchies}
\label{app:comparison}

In the main text, we have introduced three different hierarchies of moment-based conditions for detecting Wigner negativity, namely: (i) the LP hierarchy, (ii) the LC hierarchy, and (iii) the Hankel-matrix hierarchy. While all these conditions follow from positivity of the Wigner function, they capture different structural aspects of the moment sequence. It is natural to ask how they are related to one another and which hierarchy provides the strongest detection power. In this section, we establish the precise analytical hierarchy among these conditions.

\vspace{5mm}
\noindent \textbf{Proposition 2.} \textit{If the Hankel matrix $H_n(\boldsymbol{w})$ is positive semidefinite, then the moments satisfy the log-convexity condition}
\begin{equation}
LC(k) =  w_{k-1}w_{k+1} -w_k^2 \geq 0,
\qquad \forall \; k\le 2n .
\end{equation}

\begin{proof}
Since $H_n(\boldsymbol{w})$ is positive semidefinite, every principal submatrix of $H_n(\boldsymbol{w})$ must also be positive semidefinite. Consider the $2\times2$ principal submatrix
\begin{equation}
\begin{bmatrix}
w_{k-1} & w_k \\
w_k & w_{k+1}
\end{bmatrix}, \qquad k\leq2n.
\end{equation}
Positivity of this matrix implies that its determinant is non-negative, which gives $w_{k-1}w_{k+1}-w_k^2\ge0$. 
\end{proof}

This implies that, Hankel matrix condition is stronger than the LC condition when the maximum available moment order is fixed. Consequently, violation of any such LC condition, guaranties the violation of the Hankel condition. However the opposite is not true, which is supported by the examples studied in this letter. Next, we compare the LC hierarchy with the LP hierarchy.

\vspace{5mm}
\noindent \textbf{Proposition 3.} \textit{For a fixed maximum accessible moment order, satisfaction of the LC condition guarantees satisfaction of the corresponding LP condition.}
\begin{proof}
The proof proceeds by mathematical induction. The minimal condition for positivity in the LC hierarchy is $w_1w_3 - w_2^2 \geq 0$. Since $w_1=1$, this reduces to $w_3\ge w_2^2$, which is precisely the $LP(3)$ condition, which means at the lowest order, satisfaction of LC implies satisfaction of LP. Now, taking $(n-1)$th power in the $LC(n)$ condition, we get
\begin{align} \label{eq:LC_LP_proof}
w_{n+1}^{\,n-1}w_{n-1}^{\,n-1}\ge w_n^{\,2n-2}.
\end{align}

Now let us assume that the LP condition at order $n$ holds, i.e. 
\begin{equation}
w_n^{\,n-2}\ge w_{n-1}^{\,n-1}.
\label{eq:LP_assumption}
\end{equation}

Applying Eq.~\eqref{eq:LP_assumption} on the right hand side of Eq.~\eqref{eq:LC_LP_proof},
\begin{align}
& \;\; w_{n+1}^{\,n-1}w_{n-1}^{\,n-1}
\ge
 w_n^{2n-2}=w_n^n \cdot  w_n^{\,n-2} \nonumber \\
 \Rightarrow & \;\; w_{n+1}^{\,n-1}w_{n-1}^{\,n-1} \ge w_n^n w_{n-1}^{\,n-1} \nonumber \\
 \Rightarrow & \;\; w_{n+1}^{\,n-1}\ge w_n^n \; ,\nonumber
\end{align}
which is exactly the LP condition at order $n+1$. Thus the proof is completed using mathematical induction.
\end{proof}

Combining these results, we obtain the hierarchy of strength for a fixed maximum accessible moment order,
\begin{equation}
\text{Hankel} \Longrightarrow \text{LC} \Longrightarrow \text{LP}.
\end{equation}

To justify the relative strength of the hierarchies, we emphasize that the LP conditions involve only two consecutive moments and therefore probe only limited information about the full moment sequence. In contrast, the LC conditions involve correlations among three neighboring moments and are therefore strictly stronger.
Similarly, the LC conditions arise only from positivity of the $2\times2$ principal minors of the Hankel matrices. Positivity of the full Hankel matrices imposes significantly stronger constraints, since it incorporates all moments up to a given order simultaneously. Consequently, there exist moment sequences satisfying all LC conditions while still violating positivity of the corresponding Hankel matrix.

Therefore, when the maximum accessible moment order is fixed, the Hankel hierarchy always provides the strongest detection criterion, followed by the LC hierarchy and finally the LP hierarchy. This explains the behavior observed in the examples studied in the main text, where the Hankel conditions consistently detect the largest set of Wigner-negative states.

\section{E. Quantifying Wigner negativity using moments}\label{SM_E}

A natural way to quantify Wigner negativity is through the absolute volume of the Wigner function. For a $N$-mode state with Wigner function $W$, the negativity measure is given by
\begin{equation}
    \mathcal{N}(W) = \int |W(\vec{x},\vec{p})| \, d^N x \, d^N p .
\end{equation}
For all quantum states with a positive Wigner function, one has $\mathcal{N}(W) = 1$, since $W$ is normalized. In contrast, whenever $W$ attains negative values,  the positive and negative contributions no longer cancel under the absolute value and thus $\mathcal{N}(W) > 1$. Hence, $\mathcal{N}(W)$ serves as a faithful indicator of Wigner negativity. In analogy with logarithmic negativity in entanglement theory \cite{vidal2002computable}, the logarithmic measure of Wigner negativity can be defined as \cite{PhysRevA.98.052350}
\begin{equation}
    L(W) = \log \left( \int |W| \, d^N x \, d^N p \right) = \log(\tilde{w}_1),
\end{equation}
By construction, $L(W)=0$ for all Wigner-positive states and increases monotonically with the amount of negativity. However, evaluating this quantity requires pointwise knowledge of $|W|$, making it experimentally demanding as it generally necessitates full state tomography. This motivates the development of alternative quantifiers that can be estimated efficiently from limited moment data, without requiring full reconstruction of the state.

To connect these quantifiers with the moment-based framework developed in this letter, it is natural to consider moments of the non-negative function $|W|$. We therefore define the absolute-valued Wigner moments as
\begin{equation}
    \tilde{w}_m = \int |W(\vec{x},\vec{p})|^m \, d^N x \, d^N p \;, \qquad m \in \mathbb{N}.
\end{equation}

Since $|W|$ is non-negative, but generally not normalized whenever the Wigner function contains negative regions, it is natural to associate with it the normalized probability distribution
\begin{equation}
    f(\vec{x},\vec{p}) = \frac{|W(\vec{x},\vec{p})|^2}{\tilde{w}_2},
\end{equation}
which satisfies $f\ge 0$ and $\int f =1$.

We define the R\'enyi entropy of order $r$ associated with this distribution as
\begin{equation}
    {R}_r(W) = \frac{1}{1-r} \log \left( \int f(\vec{x},\vec{p})^r \, d^N x \, d^N p \right).
\end{equation}
Using the definition of absolute-valued Wigner moments, this can be written as
\begin{equation}
    {R}_r(W) = \frac{1}{1-r} \log \left( \frac{\tilde{w}_{2r}}{\tilde{w}_2^r} \right).
\end{equation}
We now define the following quantity:

\begin{align}
    L_r(W)
    &= \frac{1}{2-r} \left[ \log(\tilde{w}_r) + (1 - r)\log(\tilde{w}_2) \right] \\
     &= \frac{1}{2} \left( {R}_{r/2}(W) - \log(\tilde{w}_2) \right).
\end{align}

We note that $\tilde{w}_r = w_r$ whenever $r$ is an even integer. This also clarifies the choice of $\tilde{w}_2$ in the definition of $f(\vec{x},\vec{p})$, as it is the lowest-order nontrivial absolute Wigner moment that can be directly estimated within the framework developed in this letter. As a result, $L_r(W)$ can be efficiently evaluated for any even $r$. Importantly, the required moments are extracted from the same measurement data used for negativity detection, enabling simultaneous detection and quantification without additional experimental overhead.

Before proving the results stated in the main text, we first establish a standard property of the R\'enyi entropies associated with the normalized distribution $f(\vec{x},\vec{p})$. This observation will play a central role in identifying the optimal member of the moment-based quantifier hierarchy.

\begin{lemma}[Monotonicity of R\'enyi entropies] \label{lemma:renyi}
Let $R_r(W)$ denote the R\'enyi entropy associated with the probability distribution $f(\vec{x},\vec{p})$. Then, for $r<s$,
\begin{equation}
R_r(W)\ge R_s(W).
\end{equation}
\end{lemma}

\begin{proof}
Let
\begin{equation}
S(r)=\int f(\vec{x},\vec{p})^r\,d^Nx\,d^Np,
\end{equation}
so that
\begin{equation}
R_r(W)=\frac{\log S(r)}{1-r}.
\end{equation}
Defining $p_r(\vec{x},\vec{p}) = {f(\vec{x},\vec{p})^r}/{S(r)}$, which is a probability distribution by definition, one finds after some straightforward algebra that

\begin{equation}
\frac{d}{dr}R_r(W)
=-\frac{1}{(1-r)^2}
D\!\left(p_r\,|\,f\right),
\end{equation}
where
\begin{equation}
D(p_r|f)
=
\int p_r
\log\!\left(\frac{p_r}{f}\right)
d^Nx\,d^Np
\end{equation}
is the Kullback--Leibler divergence. Since $D(p_r|f)\ge0$, it immediately follows that
\begin{equation}
\frac{d}{dr}R_r(W)\le0.
\end{equation}
Hence $R_r(W)$ is non-increasing in $r$, proving the claim.
\end{proof}

\subsection*{Proof of Theorem~5}

Using the definitions of $L_r(W)$ and $L(W)$,
\begin{equation}
L_r(W)-L(W)
=
\frac{1}{2-r}
\log\!\left(
\frac{\tilde{w}_r\,\tilde{w}_1^{\,r-2}}
{\tilde{w}_2^{\,r-1}}
\right).
\end{equation}
Since $r>2$, we have $2-r<0$. Therefore it suffices to prove
\begin{equation}
\tilde{w}_2^{,r-1}
\le
\tilde{w}_1^{,r-2}\tilde{w}_r.
\end{equation}

Applying the interpolation inequality for $L_p$ norms,
\begin{equation}
|W|_2
\le
|W|_1^{\frac{r-2}{2(r-1)}}
|W|_r^{\frac{r}{2(r-1)}},
\end{equation}
and raising both sides to the power $2(r-1)$ yields
\begin{equation}
|W|_2^{\,2(r-1)}
\le
|W|_1^{\,r-2}
|W|_r^{\,r}.
\end{equation}
Using
\begin{equation}
|W|_1=\tilde{w}_1,\qquad
|W|_2^2=\tilde{w}_2,\qquad
|W|_r^r=\tilde{w}_r,
\end{equation}
we obtain
\begin{equation}
\tilde{w}_2^{\,r-1}
\le
\tilde{w}_1^{\,r-2}\tilde{w}_r,
\end{equation}
which proves
\begin{equation}
L_r(W)\le L(W).
\end{equation}

\subsection*{Proof of Corollary 2}

If $W(\vec{x},\vec{p})\ge0$, then
$\tilde{w}_1
=
1$
and therefore $L(W)=0$. Then Theorem~4 immediately gives
\begin{equation}
L_r(W)\le L(W)=0.
\end{equation}

Conversely, if $L_r(W)>0$, then Theorem~4 implies $L(W)>0$, which can only occur when the Wigner function attains negative values. This completes the proof.

\subsection*{Proof of Proposition~1}

From
$L_r(W)
=
\frac12
\left(
R_{r/2}(W)-\log\tilde{w}_2
\right)$,
together with Lemma~\ref{lemma:renyi}, it follows immediately that
\begin{equation}
L_r(W)\le L_s(W),
\qquad r>s,
\end{equation}
because the term $\tilde{w}_2$ is independent of $r$. Hence the hierarchy $\{L_r(W)\}_{r>2}$ is monotonically decreasing. Restricting to experimentally accessible even orders yields
\begin{equation}
L_4(W)\ge L_6(W)\ge L_8(W)\ge \cdots.
\end{equation}
Since every $L_r(W)$ provides a certifiable lower bound to the logarithmic Wigner negativity, the largest member of this hierarchy furnishes the tightest bound. Hence, $L_4(W)$ is the optimal experimentally accessible moment-based quantifier.
\qed

\subsection*{Other properties of $L_r(W)$}
We now discuss some other fundamental properties of the quantity $L_r(W)$, which further clarify its role as a quantitative measure of Wigner negativity. 

\begin{enumerate}
    \item[(i)] \textbf{Invariance under Gaussian unitaries.}
    For any symplectic transformation $S$ in phase space,
    \begin{equation}
        L_r(W \circ S) = L_r(W).
    \end{equation}
\begin{proof}
    Under a Gaussian unitary, the Wigner function transforms according to
\begin{equation}
    W(\vec{x},\vec{p}) \mapsto W(S^{-1}(\vec{x},\vec{p})),
\end{equation}
where $S$ is a symplectic transformation preserving phase-space volume. Since the Lebesgue measure is invariant under such transformations, it follows that
\begin{equation}
    \tilde{w}_k = \int |W(\vec{x},\vec{p})|^k \, d^N x \, d^N p
\end{equation}
remains unchanged for all $k$. As $L_r(W)$ depends only on $\tilde{w}_r$ and $\tilde{w}_2$, it is therefore invariant under Gaussian unitaries.
\end{proof}

    \item[(ii)] \textbf{Additivity.}
    For product states,
    \begin{equation}
        L_r(W_1 \otimes W_2)
        =
        L_r(W_1) + L_r(W_2).
    \end{equation}

    \begin{proof}
        For a product state, the Wigner function factorizes as
\begin{equation}
    W_{12}(\vec{x}_1,\vec{p}_1,\vec{x}_2,\vec{p}_2)
    =
    W_1(\vec{x}_1,\vec{p}_1)\, W_2(\vec{x}_2,\vec{p}_2).
\end{equation}
Consequently,
\begin{equation}
    \tilde{w}_k(W_1 \otimes W_2)
    =
    \int |W_1 W_2|^k
    =
    \tilde{w}_k(W_1)\,\tilde{w}_k(W_2),
\end{equation}
for all $k$. Substituting this multiplicative relation into the definition of $L_r(W)$ and using the additivity of the logarithm yields
\begin{equation}
    L_r(W_1 \otimes W_2)
    =
    L_r(W_1) + L_r(W_2).
\end{equation}
    \end{proof}
\end{enumerate}

\subsection*{Illustrative examples}

To demonstrate the effectiveness of the proposed moment-based quantifiers, we evaluate them for two representative classes of non-Gaussian states, as illustrated in Fig.~\ref{fig:measure}. 
The first example we consider here is the $N00N$ state, which is a two-mode quantum state where all $N$ photons can be found in either one of the two modes $a$ and $b$, while the other remains vacant. It is a maximally path-entangled number state of the form
\begin{equation}
    |\psi, N\rangle = \frac{1}{\sqrt{2}}\left(|N\rangle_{a} |0\rangle_{b} + e^{i\phi}|0\rangle_{a} |N\rangle_{b}\right).
\end{equation}

For the choice $\phi = \pi$, the corresponding Wigner function in terms of the dimensionless quadratures $(x_1,p_1)$ and $(x_2,p_2)$ is given by \cite{agarwal2012quantum}
\begin{equation}
    \begin{split}
         W_{N}(x_1,p_1,x_2,p_2) = \frac{1}{2\pi^{2}N!} e^{-(x_1^2+x_2^2+p_1^2+p_2^2)} &\times \Big[ -2^{N} \big( (x_1 + ip_1)^{N} (x_2 - ip_{2})^{N}  +  (x_1 - ip_1)^{N} (x_2 + ip_{2})^{N} \big)  \\
         &+(-1)^{N} N!\big( \mathcal{L}_{N}[2(x_1^2+p_1^2)]  + \mathcal{L}_{N}[2(x_2^2+p_2^2)]\big)\Big],
    \end{split}
\end{equation}
where $\mathcal{L}_N(x)$ denotes the Laguerre polynomial. This state exhibits Wigner negativity for all $N \geq 1$.

The second example we consider here is the photon added coherent state, which is obtained by repeated application of the photon creation operator on a coherent state, which initially has properties like a classical field. The state is defined by

\begin{align}
    \ket{\alpha, m} = \frac{ {a^{\dagger}}^m \ket{\alpha} }{ (\langle\alpha|a^m {a^\dagger}^m | \alpha\rangle)^{1/2}}.
\end{align}
Here $\ket{\alpha}$ is a coherent state, with $\alpha \in\mathbb{C}$ and m is an integer. The Wigner function for this state is given by \cite{agarwal_1991_photon_added},
\begin{align}
    W(z) = \frac{2 (-1)^m \mathcal{L}_m(|2z-\alpha|^2)) }{\pi \mathcal{L}_m(-|\alpha|^2)} \exp( -2|z-\alpha|^2 ), \;\; z\in \mathbb{C}
\end{align}

For all of these families of states, the logarithmic Wigner negativity $L(W)$ increases monotonically with the excitation number, indicating the enhancement of nonclassical features (see Fig.~\ref{fig:measure}). The quantities ${L}_4(W)$ and ${L}_6(W)$ consistently act as lower bounds to ${L}_{N}(W)$, in accordance with our theoretical framework. Furthermore, the ordering ${L}_4(W) \geq {L}_6(W)$ is observed across all cases, with lower-order moments providing tighter bounds. These observations confirm that the proposed moment-based measures successfully capture both qualitative trends and quantitative aspects of Wigner negativity.
At the same time, there exist Wigner-negative states for which \(L_r(W)\) does not yield a strictly positive value. This is because the quantities \(L_r(W)\) are constructed from only a finite set of low-order Wigner moments and therefore capture only partial information about the underlying phase-space distribution. The resulting loss of sensitivity is the natural trade-off for obtaining experimentally accessible and efficiently computable lower bounds on Wigner negativity.

\begin{figure*}[t]
    \centering
    \begin{subfigure}{0.44\textwidth}
        \centering
        \includegraphics[width=\textwidth]{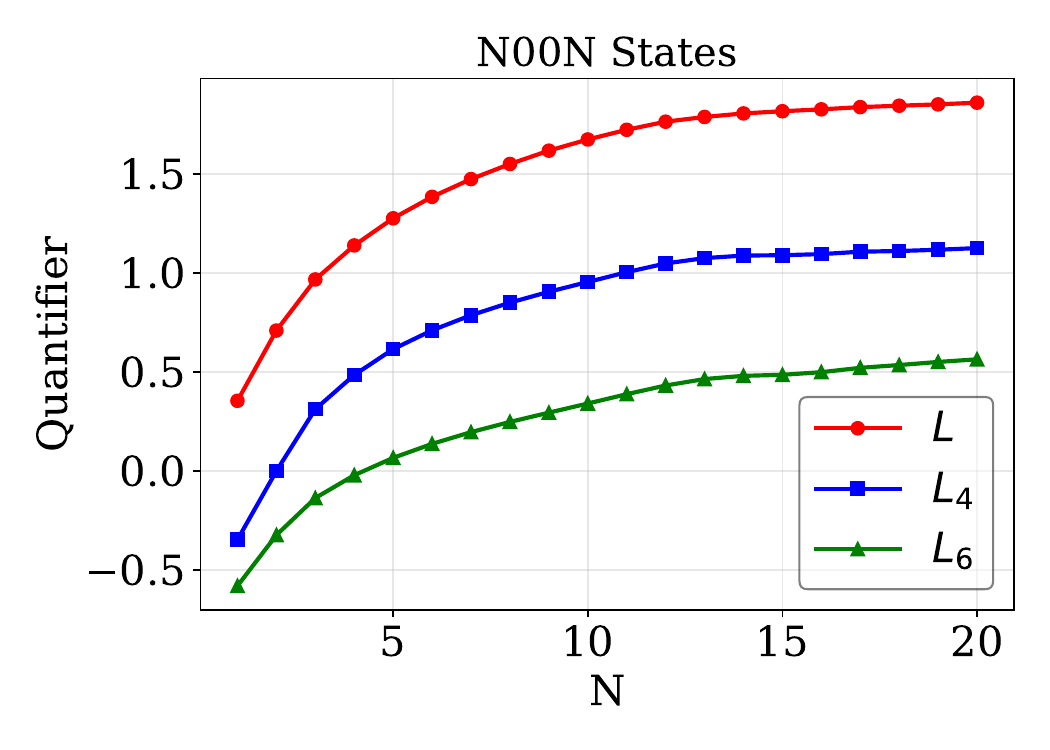}
        \caption{}
        \label{fig:L_NOON}
    \end{subfigure}
    \begin{subfigure}{0.44\textwidth}
        \centering
        \includegraphics[width=\textwidth]{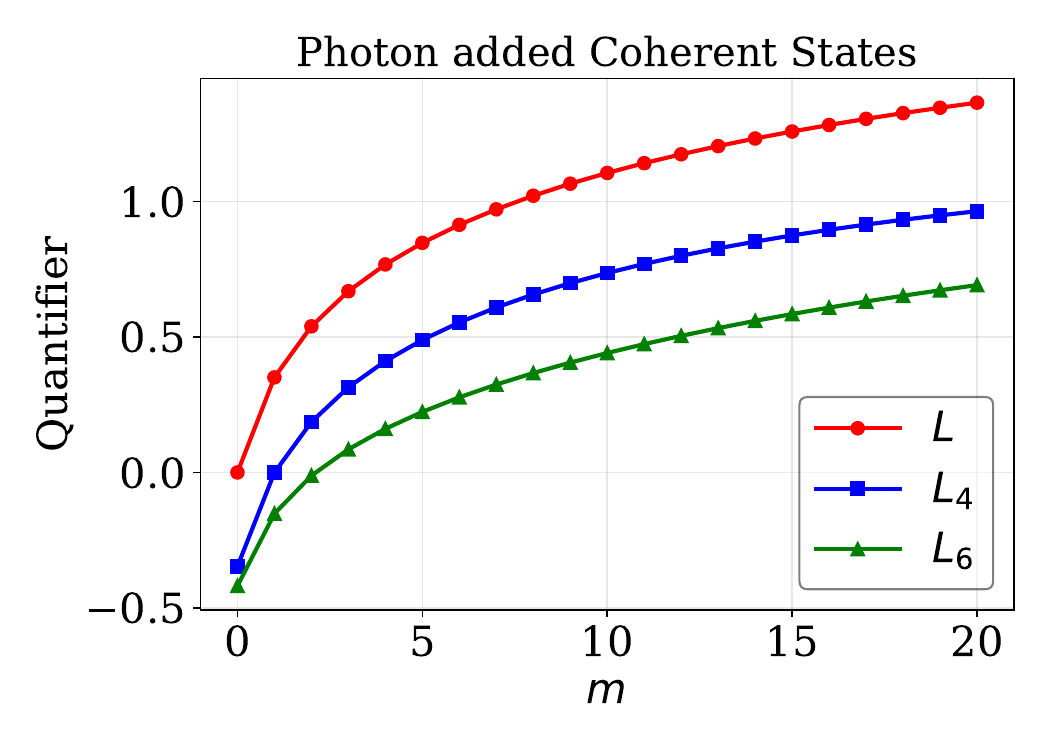}
        \caption{}
        \label{fig:L_PACS}
    \end{subfigure}
\caption{
\small
\justifying
Quantification of Wigner negativity using moment-based measures for representative non-Gaussian states.
(a) $N00N$ states $\{ \ket{\psi, N} \}$: The logarithmic Wigner negativity ${L}(W)=\log\!\int |W|$ and the corresponding moment-based quantities ${L}_4(W)$ and ${L}_6(W)$ are plotted as functions of the particle number $N$. (b) Photon added coherent states $\ket{\alpha,m}$: The Wigner negativity quantifiers are plotted as a function of $m$, the number of photons added in a coherent state with $\alpha=0.1$. The hierarchy among the moment-based measures is maintained in both the plots, demonstrating their consistency as lower bounds.
}
\label{fig:measure}
\end{figure*}

\section{F. Applications: Detection of Quantum Entanglement via Wigner Moments}\label{SM_F}
Direct links between entanglement and Wigner negativity have been recently
uncovered \cite{Zaw2024_bipartite_entanglement, zaw2025_GME_1}, shedding light on the structure of quantum correlations. In particular, negativities in appropriately constructed Wigner functions, such as reduced or collective phase-space distributions have been shown to imply the presence of bipartite entanglement or even genuine multipartite entanglement. 
We exploit this connection to lift our moment-based nonclassicality criteria to the detection of quantum entanglement.

\subsection{Detection of bipartite entanglement}
\label{subsec:bipartite}

First we demonstrate how the moment-based criteria developed in the previous sections can be employed to detect bipartite entanglement in CV systems utilizing the recently established connections between Wigner negativity and the negativity of the partial transposed density matrix.

Let $\rho$ be an $N$-mode CV state. We consider an \emph{equal bipartition} of the modes into two subsystems,
\begin{equation}\label{eq:equal_bipartition}
\vec{a}_A := \{a_m\}_{m=1}^{N/2}, \qquad
\vec{a}_B := \{a_m\}_{m=N/2+1}^{N}.
\end{equation}
We denote the partial transpose of $\rho$ over the modes $\vec{a}_B$ as $\rho^{T_B}$.
Then entanglement across this bipartition can be characterized using the negativity of the partial transpose criterion: if the partially transposed state $\rho^{T_B}$ is not positive semidefinite, then $\rho$ is necessarily entangled. 

A powerful connection arises from a phase-space formulation of the PPT condition. In particular, it has been shown that bipartite entanglement can be inferred from negativities in suitably defined reduced Wigner functions. To describe this construction, we introduce collective modes defined componentwise as
\begin{equation}
\vec{a}_{\pm} = \frac{1}{\sqrt{2}}\left(\vec{a}_A \pm \vec{a}_B\right),
\end{equation}

Given the state $\rho$, we calculate the reduced state $\mathrm{Tr}_{-}\rho$,
obtained by tracing out all relative modes $\vec{a}_-$. The Wigner function of this reduced state, denoted $W_{\mathrm{Tr}_{-}\rho}$, is obtained from the full Wigner function $W_\rho$ via marginalization over the phase-space variables associated with $\vec{a}_-$. A key result, proven in Ref.~\cite{Zaw2024_bipartite_entanglement}, establishes that negativity of this reduced Wigner function implies violation of the PPT condition. More precisely, one has
\begin{equation}
W_{\mathrm{Tr}_{-}\rho}(\vec{\alpha}) \not\geq 0
\;\;\Longrightarrow\;\;
\rho^{T_B} \not\succeq 0,
\end{equation}
which certifies NPT entanglement across the chosen bipartition. This result provides a direct operational link between phase-space nonclassicality and bipartite entanglement.

Importantly, this approach reduces the problem of entanglement detection to identifying negativities in a reduced phase-space distribution. However, directly resolving such negativities typically requires high-resolution phase-space measurements. This motivates the use of moment-based criteria, which provide coarse-grained yet experimentally accessible signatures of Wigner negativity. Combining these observations, we obtain the following result.

\begin{proposition} \label{proposition: bipartite}
Let $\rho$ be an $N$-mode CV state with an equal bipartition as defined in Eq.~\eqref{eq:equal_bipartition}, and let $W_{\mathrm{Tr}_{-}\rho}$ denote the Wigner function of the reduced state obtained by tracing out the relative modes $\vec{a}_-$. If $W_{\mathrm{Tr}_{-}\rho}$ violates any of the moment-based criteria in Theorems~1, 2 or 3, then $\rho$ is NPT entangled.
\end{proposition}

\begin{proof}
The criteria developed in Theorems~1, 2 or 3, provide sufficient conditions for the existence of negativities in a Wigner function. Therefore, violation of any of these criteria for $W_{\mathrm{Tr}_{-}\rho}$ implies that there exists $\vec{\alpha}$ such that
\[
W_{\mathrm{Tr}_{-}\rho}(\vec{\alpha}) < 0.
\]
By the established implication between reduced Wigner negativity and partial transpose negativity, this leads to $\rho^{T_B} \not\succeq 0$. Hence, $\rho$ is NPT entangled.
\end{proof}

The above theorem demonstrates that moment-based phase-space criteria can be systematically promoted to entanglement witnesses. A key advantage of this approach is that it requires only partial information about the Wigner function encoded in a finite set of moments, rather than a full phase-space reconstruction. This makes the resulting witnesses particularly suitable for experimental platforms where phase-space measurements are available but limited.

\subsection{Detection of genuine multipartite entanglement}
\label{subsec:GME}

We now extend the above framework to the multipartite regime, where the structure of quantum correlations becomes significantly richer. In contrast to bipartite entanglement, multipartite systems admit different layers of correlations depending on how the system can be partitioned. The strongest form of such correlations is \emph{genuine multipartite entanglement} (GME), which captures the presence of inseparability across \emph{all} possible bipartitions.

Formally, an $N$-mode state $\rho$ is said to be genuinely multipartite entangled if it cannot be written as a convex mixture of states that are separable with respect to any bipartition, i.e.,
\begin{equation}
\rho \neq \sum_{(A|\bar{A})} \sum_i p_{A}^{(i)} 
\, \rho_{A}^{(i)} \otimes \rho_{\bar{A}}^{(i)}, \qquad \text{where}\;\;\; p_{A}^{(i)} \geq 0.
\end{equation}
$A (\bar{A})$ denote a subset of modes i.e., $A= \{ m_n\}_{n=1}^{M}$ (complementary subset of modes, $\bar{A} = \{m_n\}_{n=1}^{M} \setminus A$) and runs over all bipartitions $(A|\bar{A})$. Such states exhibit correlations that cannot be reproduced by classical coordination across any grouping of subsystems, and play a central role in tasks such as quantum networks, distributed computation, and multipartite quantum communication.

Detecting GME is considerably more challenging than bipartite entanglement. Standard criteria, including PPT-based methods or covariance-matrix approaches, are often either insufficient or experimentally demanding, particularly for non-Gaussian states. This motivates the search for experimentally accessible criteria that can capture genuinely multipartite correlations using limited information.

\medskip

\noindent\textit{Phase-space approach and centre-of-mass mode.}  
Phase-space methods provide a natural route to GME detection by exploiting collective degrees of freedom of multimode systems. In particular, one introduces a \emph{centre-of-mass mode} of the form
\begin{equation}
a_+ = \frac{1}{\sqrt{M}} \sum_{m=1}^M \left( y_m a_m + z_m a_m^\dagger \right),
\end{equation}
where the coefficients $\vec{y}, \vec{z} \in \mathbb{C}^M$ are choosen such that $\vec{y}\circ\vec{y}^* - \vec{z}\circ\vec{z}^* = \vec{1}$, where $[\mathbf{A}\circ\mathbf{B}]_{m,n} = [\mathbf{A}]_{m,n}[\mathbf{B}]_{m,n}$ is the elementwise product and $\vec{1} = (1,1,\dots,1)$ is a vector of ones.

Given the full state $\rho$, one finds the reduced state describing the centre-of-mass motion by tracing out all orthogonal (relative) modes,
$\mathrm{Tr}_{-}\rho$, which retains only the collective degree of freedom associated with $a_+$. The corresponding Wigner function $W_{\mathrm{Tr}_{-}\rho}(\alpha)$ can be obtained from the full Wigner function $W_\rho(\vec{\beta})$ by marginalizing over all relative coordinates.
A key result, established in recent works (see Ref.~\cite{zaw2025_GME_1}), is that negativity of a suitably processed Wigner function of the centre-of-mass mode provides a \emph{sufficient} condition for GME. Since the reduced Wigner function may undergo smoothing due to marginalization, one considers a \emph{smoothed Wigner function} defined via convolution,
\begin{equation}
\widetilde{W}_{\mathrm{Tr}_{-}\rho}(\alpha \;; \mathcal{R})
=
\int d^{2}{\beta} \;\;
W_{\mathrm{Tr}_{-}\rho}(\beta) \,
K(\alpha - \beta \,; \mathcal{R}),
\end{equation}
where $K(\alpha - \beta \,; \mathcal{R})$ is an appropriate kernel determined by auxiliary states or filtering procedures.

The central implication can be stated as follows: if the negativities present in the centre-of-mass Wigner function are sufficiently robust to persist under such smoothing, then the underlying state cannot be decomposed into biseparable components. In particular,
\begin{equation}
\exists \alpha \;:\; \widetilde{W}_{\mathrm{Tr}_{-}\rho}(\alpha \,; \mathcal{R}) < 0
\quad \Longrightarrow \quad
\rho \ \text{is GME}.
\end{equation}

This result establishes a direct and operational bridge between phase-space nonclassicality and genuine multipartite entanglement. Importantly, it shifts the problem of GME detection to identifying negativities in a \emph{single-mode} quasiprobability distribution associated with collective degrees of freedom.

\noindent\textit{Moment-based detection of GME.} While the above criterion is powerful, directly resolving negativities in $\widetilde{W}_{\mathrm{Tr}_{-}\rho}$ may still require fine-grained phase-space measurements. This is precisely where the moment-based hierarchies developed in this work become useful.
Since the $L_p$-norm, log-convexity, and Hankel-matrix criteria provide \emph{sufficient} conditions for Wigner negativity, they can be directly applied to the smoothed centre-of-mass Wigner function to yield experimentally accessible witnesses of GME. We formalize this observation in the following theorem.

\begin{proposition}
Let $\rho$ be an $N$-mode CV state, and let $\widetilde{W}_{\mathrm{Tr}_{-}\rho} (\alpha \, ;\mathcal{R})$ denote a smoothed Wigner function associated with its centre-of-mass mode. If $\widetilde{W}_{\mathrm{Tr}_{-}\rho} (\alpha \, ;\mathcal{R})$ violates any of the moment-based criteria in Theorems~1, 2 or 3, then $\rho$ is genuinely multipartite entangled.
\end{proposition}

\begin{proof}
    The proof of this proposition follows the same arguments as in the proof of Proposition.~\ref{proposition: bipartite}.
\end{proof}

This result demonstrates that moment-based phase-space criteria provide a unified and experimentally viable framework for detecting multipartite entanglement. In contrast to conventional GME witnesses, which often require full tomography or access to high-order correlations, our approach relies only on low-order moments of a single effective mode. This significantly reduces the experimental overhead, particularly in platforms where phase-space measurements are naturally available.

\subsection{Illustrative examples}

We now illustrate the above results using representative families of multimode states. In particular, we consider parametrized non-Gaussian states for which the reduced Wigner function can be computed explicitly, and compare the detection thresholds obtained from moment-based criteria with those obtained from direct Wigner negativity and the PPT condition.

\textbf{\textit{Example 1: Entangled single-photon state.}}
We consider the bipartite single-photon state
\begin{equation}
\ket{\theta}
=
\cos\theta\, \ket{1,0}
+
\sin\theta\, \ket{0,1},\;\; 0\leq \theta\leq2\pi
\end{equation}
where $\ket{1,0} = a_1^\dagger \ket{0}$ and $\ket{0,1} = a_2^\dagger \ket{0}$ denote single-photon excitations in modes $1$ and $2$, respectively. We introduce the collective modes
\begin{equation}
a_\pm = \frac{a_1 \pm a_2}{\sqrt{2}},
\end{equation}
which can be inverted as
\begin{equation}
a_1 = \frac{a_+ + a_-}{\sqrt{2}}, 
\qquad
a_2 = \frac{a_+ - a_-}{\sqrt{2}}.
\end{equation}
The vacuum state is invariant under this change of modes, i.e., $a_\pm \ket{0} = 0$. We now rewrite the basis states in terms of the new modes. Using the above relations,
\begin{align*}
\ket{1,0}
&= a_1^\dagger \ket{0}\\
&= \frac{1}{\sqrt{2}}(a_+^\dagger + a_-^\dagger)\ket{0}\\
&= \frac{1}{\sqrt{2}}\big(\ket{1}_+ \ket{0}_- + \ket{0}_+ \ket{1}_-\big), \\
\ket{0,1}
&= a_2^\dagger \ket{0}\\
&= \frac{1}{\sqrt{2}}(a_+^\dagger - a_-^\dagger)\ket{0}\\
&= \frac{1}{\sqrt{2}}\big(\ket{1}_+ \ket{0}_- - \ket{0}_+ \ket{1}_-\big),
\end{align*}
where we have used the definitions 
\begin{align*}
    \ket{1}_+ = a_+^\dagger \ket{0}, \quad \ket{1}_- = a_-^\dagger \ket{0}.
\end{align*}

\noindent
Substituting these expressions into $\ket{\theta}$, we obtain
\begin{align}
\ket{\theta}
&= \frac{1}{\sqrt{2}}
\Big[
\cos\theta \big(\ket{1}_+ \ket{0}_- + \ket{0}_+ \ket{1}_-\big) \nonumber +
\sin\theta \big(\ket{1}_+ \ket{0}_- - \ket{0}_+ \ket{1}_-\big)
\Big] \nonumber \\
&= \frac{1}{\sqrt{2}}
\Big[
(\cos\theta + \sin\theta)\, \ket{1}_+ \ket{0}_- \nonumber  +
(\cos\theta - \sin\theta)\, \ket{0}_+ \ket{1}_-
\Big].
\end{align}
Defining
\begin{equation}
c_+ = \frac{\cos\theta + \sin\theta}{\sqrt{2}},
\qquad
c_- = \frac{\cos\theta - \sin\theta}{\sqrt{2}},
\end{equation}
the state can be written compactly as
\begin{equation}
\ket{\theta}
=
c_+ \ket{1}_+ \ket{0}_-
+
c_- \ket{0}_+ \ket{1}_-.
\end{equation}

\medskip

\noindent
We now compute the reduced state obtained by tracing out the $-$ mode:
\begin{equation}
\rho_+ = \tr_- \big( \ket{\theta}\bra{\theta} \big).
\end{equation}
Expanding the density matrix,
\begin{align}
\ket{\theta}\bra{\theta}
&=
|c_+|^2 \ket{1}_+\bra{1} \otimes \ket{0}_-\bra{0}
+
|c_-|^2 \ket{0}_+\bra{0} \otimes \ket{1}_-\bra{1} \nonumber
+ c_+ c_-^* \ket{1}_+\bra{0} \otimes \ket{0}_-\bra{1}
+ c_+^* c_- \ket{0}_+\bra{1} \otimes \ket{1}_-\bra{0}.
\end{align}
Taking the partial trace over the $-$ mode we get,
\begin{equation}
\rho_+
=
|c_+|^2 \ket{1}_+ \bra{1}
+
|c_-|^2 \ket{0}_+\bra{0}.
\end{equation}

\medskip

\noindent
The Wigner function of the reduced state is therefore a convex combination of the vacuum and single-photon Wigner functions,
\begin{equation}
W_{\tr_-\rho}(x_+,p_+)
=
|c_-|^2 W_0(x_+,p_+)
+
|c_+|^2 W_1(x_+,p_+),
\end{equation}
which is equivalent to the example of mixed Fock states presented in the main text, with $\lambda = |c_-|^2$. Hence $W_{\tr_-\rho}(x_+,p_+)$ is Wigner negative whenever,
\begin{align}
    |c_-|^2<0.5 \; &\Rightarrow \sin (2\theta)>0 \nonumber\\
    & \Rightarrow \theta \in \left(0,\frac{\pi}{2}\right) \cup \left(\pi,\frac{3\pi}{2}\right)
\end{align}  Hence, for a range of $\theta$, the reduced Wigner function exhibits negativity, which, by Proposition~\ref{proposition: bipartite}, certifies the presence of bipartite entanglement in the original state, and that entire range of $\theta$ can be detected by the moment based criteria (see Table.~1 of the main text).

\textbf{\textit{Example 2: Tripartite Dicke state.}}
As an illustrative example of the detection of GME using Wigner moments, we consider the tripartite two-excitation Dicke state
\begin{equation}
|D_2^3\rangle
=\frac{1}{\sqrt3}
\Bigl(
|110\rangle
+
|101\rangle
+
|011\rangle
\Bigr).
\end{equation}

Following Ref.~\cite{zaw2025_GME_2}, we construct the reduced Wigner function using the kernel specified by $\mathcal R={\ket{1}}$.
The resulting reduced Wigner function is
\begin{equation}
\widetilde W_{D_2^3}(\alpha;{\ket{1}})
=\frac{e^{-\frac32|\alpha|^2}}{32\pi}
\left(
81|\alpha|^6
-234|\alpha|^4
+
216|\alpha|^2
-16
\right).
\end{equation}

First we calculate the Hankel matrix of order $1$, whose determinant evaluates to
\begin{equation}
\det(H_1)
= 5.13\times10^{-3}>0.
\end{equation}
Therefore, the lowest-order moment condition is inconclusive for this state. Proceeding to the next level of the hierarchy, we find

\begin{equation}
\det(H_2)
=-2.23\times10^{-7}<0.
\end{equation}
Hence, the GME of \(|D_2^3\rangle\) is therefore certified.

\bibliography{wignernegativity}
\end{document}